\documentclass[10pt,journal,compsoc]{IEEEtran}

\usepackage[cmex10]{amsmath}
\usepackage{amssymb}
\usepackage{amsthm}
\DeclareMathOperator*{\cov}{Cov} 
\DeclareMathOperator*{\var}{Var}	
\DeclareMathOperator*{\bias}{Bias}	
\DeclareMathOperator*{\mse}{MSE}	
\usepackage{dsfont} 
\usepackage{paralist} 
\usepackage[multiple]{footmisc} 

\newtheorem{theorem}{Theorem}

\newtheorem{corollary}[theorem]{Corollary}

\newtheorem{remark}{Remark} 

\newtheorem*{problemdefinition}{Problem Definition} 

\usepackage[ruled,vlined]{algorithm2e}

\usepackage{enumerate} 

\usepackage{graphicx}
\usepackage[outdir=./]{epstopdf}
\usepackage{caption}
\usepackage{subcaption}

\usepackage{tabulary}
\usepackage{booktabs}

\newcommand{\samplingbudget}{|S|}   
\newcommand{\estIP}{I^{{\samplingbudget}}}  
\newcommand{\estRW}{T^{{\samplingbudget}}_{RW}}  
\newcommand{\estFN}{T^{{\samplingbudget}}_{FN}}   
\newcommand{\estUN}{T^{{\samplingbudget}}_{UN}} 
\newcommand{\truevalue}{\bar{f}}   

\ifCLASSOPTIONcompsoc
  \usepackage[nocompress]{cite}
\else
  \usepackage{cite}
\fi
\ifCLASSINFOpdf
\else
\fi
\begin{document}
%
\title{``What Do Your Friends Think?": \\Efficient Polling Methods for Networks Using Friendship Paradox}
%
%
%
%

\author{Buddhika~Nettasinghe,~\IEEEmembership{Student~Member,~IEEE}
	and~Vikram~Krishnamurthy,~\IEEEmembership{Fellow,~IEEE}
\IEEEcompsocitemizethanks{\IEEEcompsocthanksitem Authors are with the School of Electrical and Computer Engineering, Cornell~University.\protect\\
E-mail:  \{dwn26, vikramk\}@cornell.edu.\IEEEcompsocthanksitem This material is based upon work supported, in part by, the U. S. Army Research Laboratory and the U. S. Army Research Office under grants 12346080 and W911NF-19-1-0365 and, National Science Foundation under grant 1714180.}\\
}

\IEEEtitleabstractindextext{%
\begin{abstract}
  This paper deals with randomized polling of a social network. In the case of forecasting the outcome of an election between two candidates A and B, classical intent polling asks randomly sampled individuals: \textit{who will you vote for?} Expectation polling asks: \textit{who do you think will win?} In this paper, we propose a novel neighborhood expectation polling (NEP) strategy that asks randomly sampled individuals: \textit{what is your estimate of the fraction of votes for A?} Therefore, in NEP, sampled individuals will naturally look at their neighbors (defined by the underlying social network graph) when answering this question.  Hence, the mean squared error (MSE) of NEP methods rely on selecting the optimal set of samples from the network. To this end, we propose three NEP algorithms for the following cases: (i)~the social network graph is not known but, random walks (sequential exploration) can be performed on the graph (ii) the social network graph is unknown. For both cases, algorithms based on a graph theoretic consequence called \textit{friendship paradox} are proposed. Theoretical results on the dependence of the MSE of the algorithms on the properties of the network are established. Numerical results on real and synthetic data sets are provided to illustrate the performance of the algorithms. 
\end{abstract}

\begin{IEEEkeywords}
opinion polling, election forecasting, expectation polling, friendship paradox, variance reduction, stochastic ordering, degree distribution, graph sampling,  social networks, social sampling
\end{IEEEkeywords}}

\maketitle

\IEEEdisplaynontitleabstractindextext

%
\IEEEpeerreviewmaketitle

\IEEEraisesectionheading{\section{Introduction}
\label{sec:introduction}}
\IEEEPARstart{T}his paper deals with randomized polling of a social network with a possibly unknown structure. In the case of forecasting the outcome of an election between two candidates A and B, classical intent polling asks uniformly sampled individuals: \textit{who will you vote for?} Expectation polling asks: \textit{who do you think will win?} In this paper, we propose a novel neighborhood expectation polling strategy that asks non-uniformly sampled individuals: \textit{what is your estimate of the fraction of votes for A?} Next, we formally define the problem, explain the solution approach and the related work that motivates it.

Consider a social network represented by an undirected graph $G = (V,E)$ where, each node $v\in V$ has a label $f(v) \in \{0,1\}$.  A pollster can query a total of ${\samplingbudget}$ (called the sampling budget) number of individuals from this social network. 
\begin{problemdefinition}
Estimate,
\begin{equation}
	\label{eq:true_value}
	\truevalue  = \frac{ \vert\{v\in V: f(v) = 1\}\vert  }{\vert V\vert} 
\end{equation} which is the fraction of nodes with label 1, with a sampling budget $\samplingbudget\ll \vert V \vert$ for the following cases:
\begin{itemize}
	\item Case 1 -  graph ${G = (V,E)}$ is not known but, the graph can be explored sequentially using a random walk
	
	\item Case 2 - graph ${G = (V,E)}$ is not known but, the set of nodes~$V$ can be uniformly sampled
\end{itemize} 
\end{problemdefinition}

We propose a class of polling methods that we call {neighborhood expectation polling (NEP)} to address the above problem\footnote{Applications of this problem include forecasting the outcome of an upcoming election\cite{tumasjan2010}, estimating the fraction of individuals infected with a disease \cite{gile2011}, estimating the number of individuals interested in buying a certain product (a market research). More specific real world examples for case~1 and case~2 are discussed in Sec.~\ref{subsec:RW_method} and Sec.~\ref{subsec:FN_method} respectively.}. In NEP, a set $S \subset V$ of individuals from the social network $G = (V,E)$ are selected and asked, 
\begin{center}
	\textit{``What is your estimate of the fraction of people with label 1?"}.
\end{center}
When trying to estimate an unknown quantity about the world, any individual naturally looks at her neighbors.  Therefore, each sampled individual $s \in S$ would provide the fraction of their neighbors $\mathcal{N}(s)$, with label $1$. In other words, the response of the individual $s\in S$ for the NEP query would be, 
\begin{align}
q(s) &= \frac{\vert\{u \in \mathcal{N}(s): f(u) = 1 \}\vert}{\vert \mathcal{N}(s)\vert}. \label{eq:response_to_query}
\end{align}
Then, the average of all the responses $\frac{\sum_{s \in S}q(s)}{\vert S \vert}$ is used as the NEP estimate of the fraction $\bar{f}$. 

\subsection{Context}
\label{subsec:context}
 {\bf Why call it NEP?} NEP takes its name from the fact that, the response $q(s)$ of each sampled individual $s\in S$ is the expected label value among her neighbors i.e. ${q(s) = \mathbb{E}\{f(U)\}}$ where, $U$ is a random neighbor of the sampled individual $s\in S$.

\begin{figure*}[]
\centering
\begin{subfigure}[!h]{0.3\textwidth}
	\centering
	\includegraphics[width=2.3in]{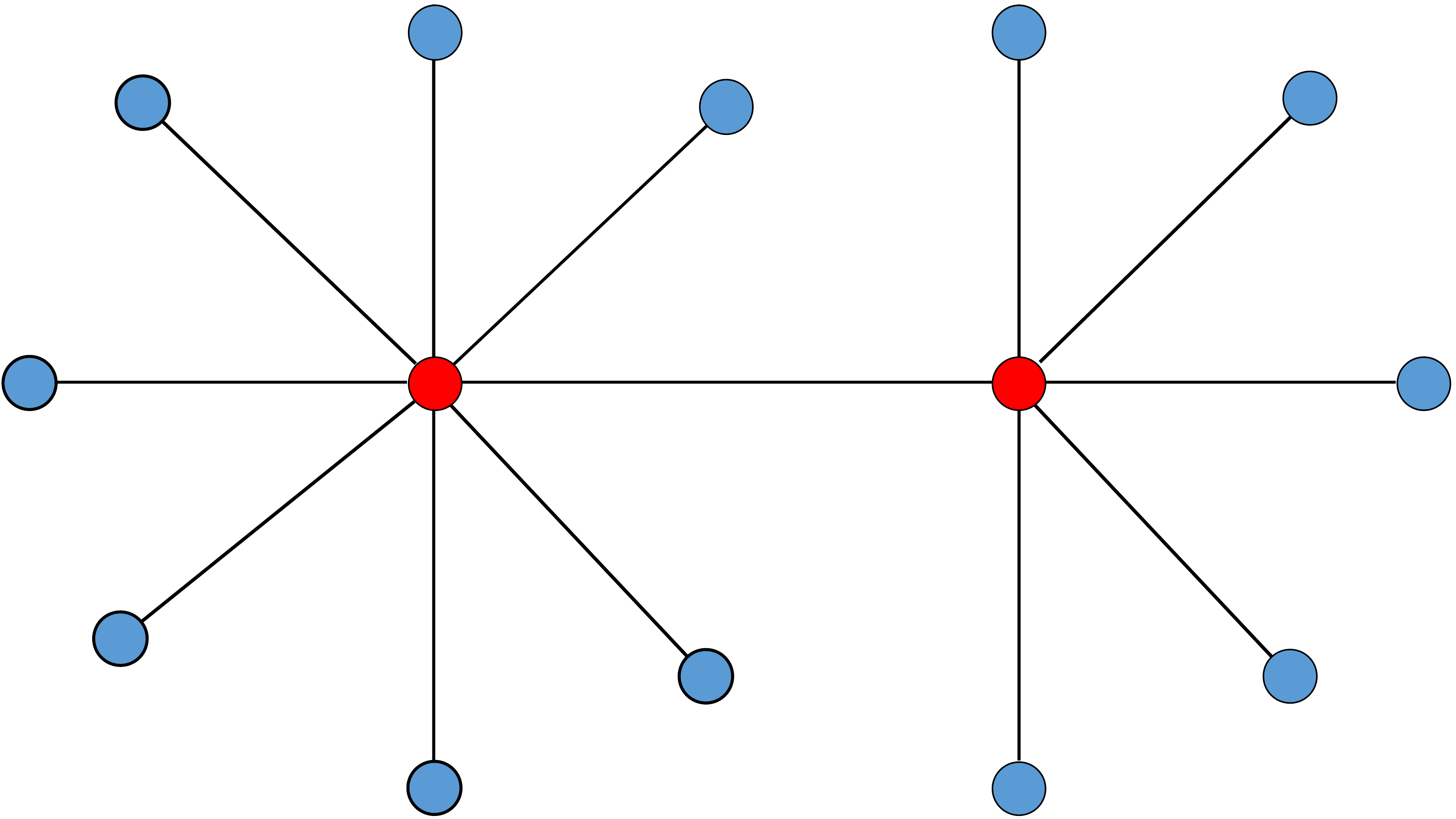}
	\caption{Network $G_1$: labels are highly correlated with the degrees of nodes }
	\label{subfig:NEP_larger_bias}
\end{subfigure}\hfill
\begin{subfigure}[!h]{0.3\textwidth}
	\centering
	\includegraphics[width=2.3in]{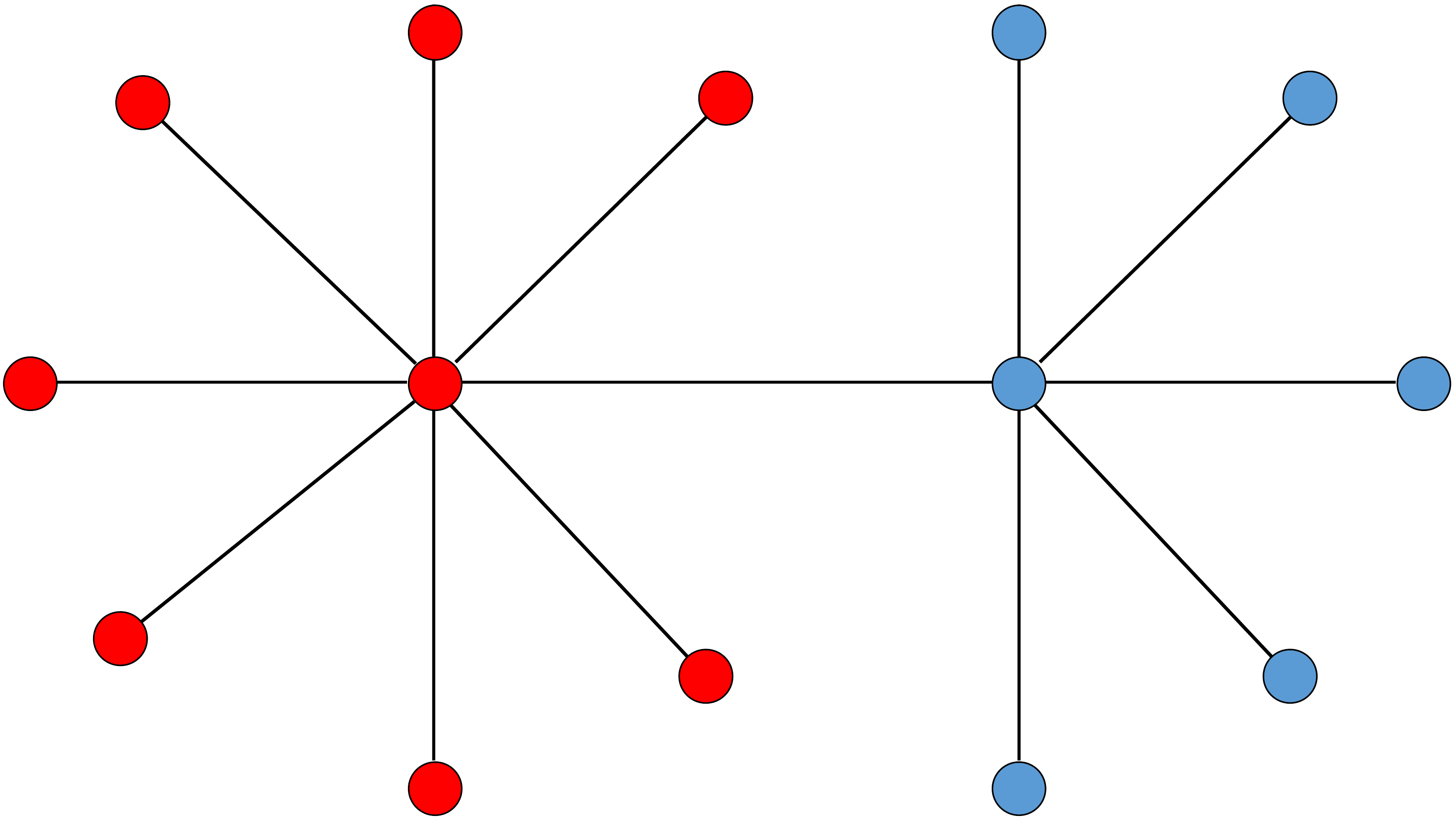}
	\caption{Network $G_2$: nodes with the same label are clustered (depicting \textit{Homophily})}
	\label{subfig:NEP_large_variance}
\end{subfigure}\hfill 
\begin{subfigure}[!h]{0.3\textwidth}
	\centering
	\includegraphics[width=2.3in]{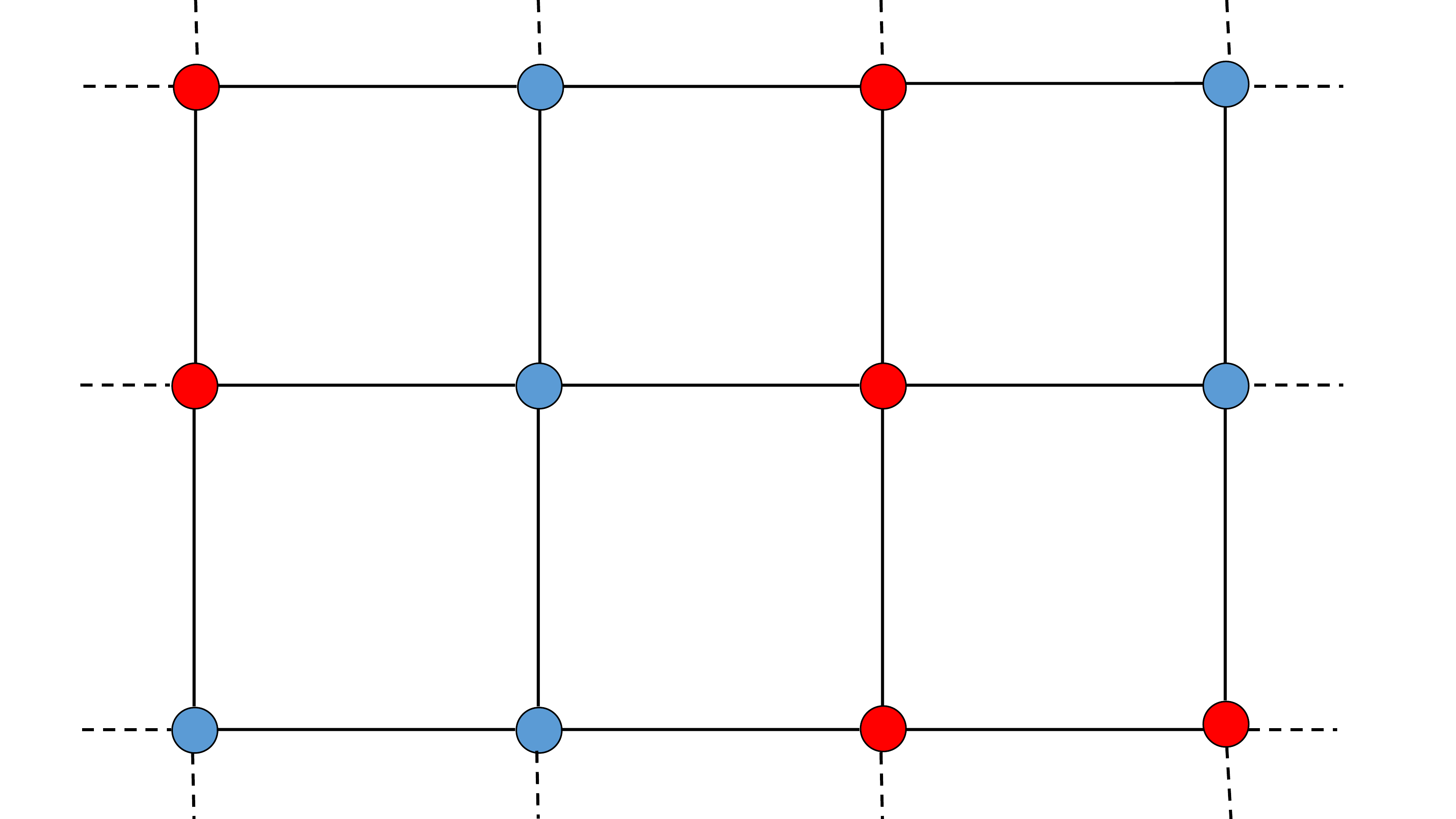}
	\caption{Network $G_3$: a large regular graph with uniformly at random assigned labels}
	\label{subfig:NEP_unbiased}
\end{subfigure}
\caption{Consider the case of uniformly sampling nodes and obtaining responses $q(s)$ of sampled nodes $s\in S$ about the fraction of red (i.e. label 1) nodes in the network. In graph $G_1$ of Fig. \ref{subfig:NEP_larger_bias}, most nodes have their only neighbor to be of color red even though most of the nodes in the network are of color blue. Hence, NEP with uniformly sampled nodes would result in a highly biased estimate in this case. In graph $G_2$ of Fig.~\ref{subfig:NEP_large_variance}, approximately half the nodes have only a red neighbor and, rest of the nodes have only a blue neighbor. Hence, NEP with uniformly sampled nodes would result in an estimate with a large variance in this case. In graph $G_3$ of Fig.~\ref{subfig:NEP_unbiased}, average of the NEP responses $q(v)$ of nodes is approximately equal to the fraction $\truevalue$ of nodes with red labels. Further, $q(v)$ does not vary largely among nodes. Hence, uniformly sampling nodes for NEP in this case would result in an accurate estimate.  Similar examples can also be found in~\cite{dasgupta2012}. This figure highlights the importance of exploiting network structure and node label distribution when sampling nodes to be used for NEP.}
\label{fig:examples_error_NEP}
\end{figure*}

\bigskip
\noindent
{\bf Why (not) use NEP?} NEP is substantially different  to classical intent polling where, each sampled individual is asked \textit{``What is your label?"}. In intent polling, the response of each sampled individual $s\in S$ is her label $f(s)$. In contrast, in NEP, the response $q(s)$ of each sampled individual $s\in S$ is a function of her neighborhood (defined by the underlying graph $G$) as well as the labels of her neighbors. Therefore, depending on  the graph $G$, function $f$ and the method of obtaining the samples $S$, NEP might produce either,
	\begin{enumerate}[I.]
		\item an estimate with a larger MSE compared to intent polling (e.g. networks in Fig. \ref{subfig:NEP_larger_bias} and Fig. \ref{subfig:NEP_large_variance} shows when uniform sampling of individuals for NEP might not work), or,
		
		\item an estimate with a smaller MSE compared to intent polling (e.g. network in Fig. \ref{subfig:NEP_unbiased} shows when uniform sampling of individuals for NEP might work)
	\end{enumerate}
These two possible outcomes highlight the importance of using the available information about the graph $G$ and the function $f$, when selecting the set $S$ of individuals in NEP. This lead us to the main results of this paper where we combine NEP with \textit{friendship paradox} (reviewed in Sec.~\ref{sec:fp}) based sampling methods to obtain statistically efficient estimates.

\begin{remark}
	\normalfont
The assumption that the graph is not fully known (case~1 and case~2 in problem definition) is applicable to most contexts that deal with large scale real world networks (including online social networks such as Facebook). This is mostly due to the fact that structures of social networks are not made available publicly by online social network network administrators and accurately estimating the network structure would incur costs (computation, memory, querying cost, etc. ) that are not feasible in the context of polling. In contrast, our methods do not rely on estimating the network structure and instead, rely on \textit{friendship paradox} based sampling method.
\end{remark}

\begin{remark}
	\normalfont
	If the graph $G = (V, E)$ is fully known, a greedy (deterministic) optimization method (similar to the one in~\cite{kempe2003}) can be used to solve the NP hard problem of finding the set $S \subset V$ of $\samplingbudget$ individuals whose collective neighborhood is largest, with a $(1- 1/e)$ approximation guarantee. However, the largest collective neighborhood does not ensure that the set $S$ of individuals would provide an accurate NEP estimate of the fraction $\truevalue$ defined in (\ref{eq:true_value}) e.g. if the sampling budget $\samplingbudget = 1$, the node with the largest collective neighborhood in the graph $G_2$ in Fig. \ref{subfig:NEP_large_variance} is the red color node with degree seven, whose NEP response (fraction of red neighbors) is $q(s) = 1$, even though $\truevalue = 4/7$. Hence, our focus is on randomized sampling methods for NEP that do not require the graph to be known. 
\end{remark}

\subsection{Main Results and Organization}
The main results of this paper are NEP algorithms for the two cases described in the problem definition and their analysis. The algorithms utilize properties related to the structure of the network to find ${\samplingbudget}$ number of samples. The analysis provides simple and intuitive conditions under which, the proposed algorithms will provide a better estimate compared to intent polling. These results can be summarized as follows.
\begin{compactitem}
	\item For case~1 and case~2, estimation algorithms are obtained by combining NEP with recent statistical results related to a phenomenon called \textit{friendship paradox}~\cite{feld1991}. Analytical results characterizing the dependence of bias, variance and MSE of estimates on the properties of the graph $G$, labels $f(v)$ of individuals $v\in V$ are obtained. These results help to identify conditions on the graph and the labels for which, friendship paradox based NEP produces a better estimate compared to intent polling and naive NEP with uniformly sampled individuals.
	
	\vspace{0.25cm}
	\item Empirical and simulation results on five real world social network datasets and synthetic datasets are provided, illustrating the performance of the proposed algorithms compared to classical methods.  These empirical and simulation results yield useful insights that complement the analytical results. 
\end{compactitem} 

\vspace{0.25cm}
\noindent
{\bf Organization:} Sec.~\ref{sec:fp} presents a review of the key results related to friendship paradox. Sec.~\ref{sec:NEP_with_FP} presents the two NEP algorithms based on the friendship paradox for case~1 and case~2, followed by their theoretical analysis in Sec.~\ref{sec:analysis}. Sec.~\ref{sec:experiments} evaluates the proposed algorithms on empirical and synthetic datasets to illustrate and compare their performances. Finally, Sec. \ref{sec:discussion} provides a discussion about the two algorithms, their theoretical and experimental evaluations and how they relate to each other. 

\vspace{0.25cm}
\noindent
{\bf Notation:} Table~\ref{tbl:notation} summarizes the parameters and variables used frequently throughout the paper. 
\begin{table}[htbp]\caption{Summary of Notation}
	\begin{center}
		\begin{tabular}{r c p{6cm} }
			\toprule
			\multicolumn{3}{c}{\underline{Network Parameters}}\\
			\multicolumn{3}{c}{}\\
			$G = (V, E)$ & $\triangleq$ & Undirected graph with set of nodes~$V$ and set of edges~$E$ \\
			$A$ & $\triangleq$ & Symmetric adjacency matrix of the graph $G$ where
			$$A(u,v)= 
			\begin{cases}
			1,& \text{if } (u,v) \in E\\
				0,              & \text{otherwise}
			\end{cases}$$\\
			$n$ & $\triangleq$ & Number of nodes~i.e.~$n = \vert V\vert$ \\
			$M$ & $\triangleq$ & Number of friends~i.e.~$M = 2\vert E \vert$\\
			$\mathcal{N}(v)$ & $\triangleq$ & The set of neighbors of a node $v\in V$ as defined by the graph $G$\\
			$d(v)$ & $\triangleq$ & Degree of node $v\in V$ i.e. $d(v) = |\mathcal{N}(v)|$\\
			$f(v)$ & $\triangleq$ & Binary label of node $v\in V$ \\
			$\truevalue$ & $\triangleq$ & Fraction of nodes with label $1$ i.e. $$\truevalue  = \frac{ \vert\{v\in V: f(v) = 1\}\vert  }{\vert V\vert} $$\\
			$q(v)$ &  $\triangleq$ & NEP response of node $v \in V$ i.e.$$q(v) = \frac{\vert\{u \in \mathcal{N}(v): f(u) = 1 \}\vert}{\vert \mathcal{N}(v)\vert}$$ \\
			$D$ & $\triangleq$ & Diagonal matrix with $D(v,v) = d(v)$\\		
			$\mathcal{A}$ & $\triangleq$ & Normalized adjacency matrix ${\mathcal{A} = D^{-\frac{1}{2}}AD^{-\frac{1}{2}}}$ \\
			\multicolumn{3}{c}{}\\
			\multicolumn{3}{c}{\underline{Random Variables, Distributions and Related Parameters}}\\
			\multicolumn{3}{c}{}\\
			$X$ & $\triangleq$ & Uniformly sampled node from set of nodes~$V$\\
		$Y$ & $\triangleq$ & Random friend: uniform sampled end of a uniformly sampled edge from~$E$\\
		$Z$ & $\triangleq$ & Random friend of a random node\\ 
		$P(k)$ & $\triangleq$ & Degree distribution which gives the probability that a random node $X$ has degree $k$ \\
		$q(k)$ & $\triangleq$ & Neighbor degree distribution that gives the probability that a random friend $Y$ has degree $k$ \\
		$e(k,k')$ & $\triangleq$ & Joint degree distribution that gives the probability that a random edge $(U,Y)$ will have nodes with degrees $d(U) = k, d(Y) = k'$ \\
		$\sigma_{k}$ & $\triangleq$ & Standard deviation of the degree $d(X)$ of a random node $X$ i.e. standard deviation of the degree distribution\\ 
		$\sigma_{f}$ & $\triangleq$ & Standard deviation of the label $f(X)$ of a random node $X$ \\
		$r_{kk}$ & $\triangleq$ & Neighbor degree correlation coefficient defined in~(\ref{eq:deg_deg_corr})\\ 
		$\rho_{kf}$ & $\triangleq$ & Degree-label correlation coefficient defined in~(\ref{eq:deg_label_corr})\\ 
				
		\multicolumn{3}{c}{}\\
		\multicolumn{3}{c}{\underline{Polling Estimates and Related Parameters}}\\
		\multicolumn{3}{c}{}\\
			${S}$ & $\triangleq$ & Set of the individuals queried by the pollster\\
			${\samplingbudget}$ & $\triangleq$ & Sampling budget~(number of individuals queried by the pollster)\\
			$N$ & $\triangleq$ & Length of Random Walk (for Algorithm~\ref{alg:RW_Sampling})\\
			$\estIP$ & $\triangleq$ & Intent polling estimate defined in~(\ref{eq:intent_polling_estimate})\\
			$\estUN$ & $\triangleq$ & Naive NEP estimate with uniformly sampled nodes defined in~(\ref{eq:naive_NEP})\\
			$\estRW$ & $\triangleq$ & NEP estimate obtained via proposed Algorithm~\ref{alg:RW_Sampling} \\
			$\estFN$ & $\triangleq$ & NEP estimate obtained via proposed Algorithm~\ref{alg:Friend_of_Node_Sampling} \\
			\bottomrule
		\end{tabular}
	\end{center}
	\label{tbl:notation}
\end{table}

\subsection{Related work}
\label{subsec:Related Work}
As described above,  in the classical intent polling\footnote{This method is called intent polling because, in the case of predicting the outcome of an election, this is equivalent to asking the voting intention of sampled individuals i.e. asking ``Who are you going to vote for in the upcoming election?'') \cite{rothschild2011}.}, a set~$S$ of nodes is obtained by uniform sampling with replacement and then, the average of their labels
\begin{equation}
\label{eq:intent_polling_estimate}
\estIP = \frac{\sum_{u\in S}f(u)}{\samplingbudget},
\end{equation} is used as the estimate (called intent polling estimate henceforth) of the fraction $\truevalue$ defined in~(\ref{eq:true_value}).  The main limitation of intent polling is that the sample size needed to achieve an $\epsilon$-~additive error is $O(\frac{1}{\epsilon^2})$~\cite{dasgupta2012}. Our work is motivated by two recently proposed methods, namely ``expectation polling" \cite{rothschild2011} and ``social sampling" \cite{dasgupta2012}, that attempt to overcome this limitation in intent polling. 

Firstly, in expectation polling \cite{rothschild2011}, each sampled individual provides an estimate of the label held by the majority of the individuals in the network (i.e.~sampled individuals answer the question ``Who do you think will win the election?"). Then, each sampled individual will look at her neighbors and provide the value held by the majority of them. This method is more efficient (in terms of sample size) compared to the intent polling method since each sampled individual now provides the putative response of a neighborhood\footnote{Intent polling and expectation polling have been considered intensively in literature, mostly in the context of forecasting elections and, it is generally accepted that expectation polling is more efficient compared to intent polling \cite{graefe2015accuracy, murr2015wisdom, graefe2014accuracy, murr2011wisdom, manski2004measuring}.}\footnote{\cite{krishnamurthy2016, krishnamurthy2014online} discuss how expectation polling can give rise to misinformation propagation in social learning and, propose Bayesian filtering methods to eliminate the misinformation propagation.}. Secondly, in social sampling\cite{dasgupta2012}, the response of each sampled individual is a function of the labels, degrees and the sampling probabilities of her neighbors. \cite{dasgupta2012} provides several unbiased estimators for the fraction $\bar{f}$ using this method and, establishes bounds for their variances.  The main limitation of social sampling method (compared to NEP) is that it requires the sampled individuals to know a significant amount of information about the underlying network. Therefore, a practical implementation of social sampling might not be feasible in settings with limited information about a very large graph. Hence, NEP can be thought of a as a method which asks a question that seeks a finer resolution compared to expectation polling and yet, simpler and intuitive compared to social sampling. 

The key idea utilized in our proposed NEP estimators for case~1 and case~2 (stated in problem definition) is the friendship paradox (detailed in Sec.~\ref{sec:fp}), which is a form of network sampling bias observed in undirected graphs. Friendship paradox has recently gained attention in several applications related to networks under the broad theme ``how network biases can be used effectively for estimation problems?". For example, \cite{nettasinghe2019maximum, eom2015tail} show how friendship paradox can be utilized for accurate estimation of a heavy tailed degree distribution, \cite{garcia2014using, christakis2010} show how friendship paradox can be used for quickly detecting a disease outbreak. Our results for the case~1 and case~2 also fall under this broad theme. 
Apart from the applications in estimation problems, friendship paradox has also been explored in the contexts of perception biases in social networks~\cite{alipourfard2019friendship, jackson2019friendship,lerman2016}, information diffusion and opinion formation~\cite{nettasinghe2019diffusions,krishnamurthy2019information,lee2019impact,bagrow2017friends}, influence maximization and stochastic seeding~\cite{chin2018evaluating, kumar2018network, lattanzi2015}, node properties other than the degrees~\cite{higham2018centrality,eom2014,momeni2016qualities} and directed social networks~\cite{alipourfard2019friendship,higham2018centrality,hodas2013}.

\section{What is Friendship Paradox?}
\label{sec:fp}
``Friendship paradox" is a graph theoretic consequence first presented in \cite{feld1991} by Scott L. Feld in 1991. The friendship paradox states, ``\textit{on average, the number of friends of a random friend is always greater than the number of friends of a random individual}". Formally:
\begin{theorem} { (Friendship Paradox \cite{feld1991})}
	\label{th:friendship_paradox_Feld}
	Consider an undirected graphs~${G = (V,E)}$. Let $X$ be a node chosen uniformly from $V$ and, $Y$~be a uniformly chosen node from a uniformly chosen edge $e\in E$. Then,
	\begin{equation}
		\mathbb{E} \{d(Y)\} \geq \mathbb{E}\{d(X)\},
	\end{equation} where, $d(X)$ and $d(Y)$ denote the degrees of $X$ and $Y$, respectively. 
\end{theorem}

In Theorem~\ref{th:friendship_paradox_Feld}, the random variable~$Y$ is called a random friend~(or a random neighbor) since it is obtained by sampling a pair of friends (i.e. an edge from the graph) uniformly and then choosing one of them by an unbiased coin flip. The intuition behind Theorem \ref{th:friendship_paradox_Feld} is as follows. Individuals with large numbers of friends appear as the friends of a large number of individuals. Hence, such popular individuals can contribute to an increase in the average number of friends of friends. On the other hand, individuals with smaller numbers of friends appear as friends of a smaller number of individuals. Hence, they cannot cause a significant change in the average number of friends of friends. Further, \cite{cao2016} shows that the original version of the friendship paradox~(Theorem~\ref{th:friendship_paradox_Feld}) is a consequence of the monotone likelihood ratio ordering between random variables $d(Y )$ and $d(X)$.

\vspace{0.25cm}
\noindent
{\bf Refinements of Friendship Paradox.} Recall that friendship paradox, in its original version given in Theorem \ref{th:friendship_paradox_Feld}, is a comparison between the degrees of a random individual~$X$ and a random friend~$Y$ (obtained by sampling an edge uniformly and then choosing one end of it by an unbiased coin flip). However, a more intuitive comparison would be the comparison of degree $d(X)$ of a random individual $X$ and the degree $d(Z)$ of a random friend $Z$ of a random individual. \cite{cao2016} develops the following important refinement of the friendship paradox which achieves this.
\begin{theorem} \cite{cao2016}
	\label{th:fosd_X_Z}
		Let ${G = (V,E)}$ be an undirected graph, $X$ be a node chosen uniformly from $V$ and, $Z$ be a uniformly chosen neighbor of a uniformly  chosen node from $V$. Then, 
		\begin{equation}
		\label{eq:fosd}
		d(Z)\geq_{fosd} d(X)
		\end{equation} 
		where, $\geq_{fosd}$ denotes the first order stochastic dominance\footnote{\label{fn:fosd}A random variable $X$ (with a cumulative distribution function~$F_X$) first order stochastically dominates a discrete random variable $Y$ (with a cumulative distribution function $F_Y$), denoted $X \geq_{fosd} Y$, if, $F_X(n) \leq F_Y(n)$, for all $n$.}. 
\end{theorem}

An immediate consequence of  Theorem \ref{th:fosd_X_Z} is,
\begin{equation}
\mathbb{E}\{d(Z)\} \geq \mathbb{E}\{d(X)\},
\end{equation}
which says that a random neighbor of a random individual has more friends than a random individual, on average (from the fact that first order stochastic dominance implies larger mean).  

With the above background, we present NEP algorithms that are based on Theorem \ref{th:friendship_paradox_Feld} and Theorem \ref{th:fosd_X_Z}.

\section{NEP Algorithms Based on Friendship Paradox}
\label{sec:NEP_with_FP}
In this section, we consider randomized methods for selecting individuals for NEP based on the concept of friendship paradox explained in Sec.~\ref{sec:fp}.

For notational reference, we first describe a naive NEP method that does not exploit the friendship paradox. 

\vspace{0.25cm}
\noindent
{\bf Naive NEP Algorithm:}
	\begin{itemize}
		\item[\it Step~1: ] Obtain a set ~$S$ of uniformly sampled nodes from $V$ and the NEP response $q(s)$~(defined in~(\ref{eq:response_to_query})) from each~$s \in S$.
		
		\item[\it Step~2: ]  Compute the naive NEP estimate of~$\truevalue$ in~(\ref{eq:true_value}) as,
		\begin{equation}
		\label{eq:naive_NEP}
		\estUN =  \frac{\sum_{s\in S} q(s)}{{\samplingbudget}} 
		\end{equation}		
	\end{itemize}

Note from the step~1  of the naive NEP method that the naive NEP estimate $\estUN$ of $\truevalue$~(fraction of nodes with label~$1$) is based on the NEP responses of uniformly sampled nodes i.e.~answers of uniformly sampled individuals to the question~``\textit{What is your estimate of the fraction of people with label 1?}" . Hence, the naive NEP algorithm exploits one's knowledge about her neighbors but does not exploit friendship paradox based sampling.  Our main contribution below is to develop NEP algorithms that exploit friendship paradox based sampling~(Sec.~\ref{subsec:RW_method} and Sec.~\ref{subsec:FN_method}) and show that they are more accurate compared to the naive NEP estimate~(\ref{eq:naive_NEP}) in terms of mean-squared error under various network structures.

\subsection{Case 1 - Sampling Friends using Random Walks}
\label{subsec:RW_method}

This subsection considers the case where the graph ${G = (V,E)}$ is not known initially, but sequential exploration of the graph is possible using multiple random walks (case~1 of problem definition) over the nodes of the graph. 

\vspace{0.25cm}
{\noindent {\bf A motivating example} for case 1 is a massive online social network where the fraction of user profiles with a certain characteristic needs to be estimated (e.g. profiles with more than ten posts about a product). Web-crawling (using random walks) approaches are widely used to obtain samples from such  massive online social networks without requiring the global knowledge of the full network graph \cite{leskovec2006,gjoka2010,ribeiro2010estimating, gjoka2011practical,mislove2007measurement}. 
	
\begin{algorithm}
	\caption{NEP with Random Walk Based Sampling}
	\label{alg:RW_Sampling}
	\DontPrintSemicolon 
	\KwIn{${\samplingbudget}$ number of samples $\{v_1, v_2,\dots, v_{\samplingbudget}\} \subset V$.}
	\KwOut{$\estRW$ which is the estimate of the fraction $\truevalue$ of nodes with label $1$.}
	
	\vspace{0.3cm}
	
	\begin{enumerate}
		\item Initialize ${\samplingbudget}$ independent random walks on the social network starting from $v_1, v_2,\dots, v_{\samplingbudget}$.

		\item Run each random walk for a $N$ steps. Then collect sample $S = \{s_1,\dots, s_{\samplingbudget}\}$ where, $s_i \in V$ is collected from $i^{th}$ random walk.

		\item Query each $s \in S$ to obtain NEP response $q(s)$ (defined in~(\ref{eq:response_to_query})) and, compute the estimate
		\begin{equation}
		\estRW =  \frac{\sum_{s \in S} q(s)}{{\samplingbudget}}. \nonumber
		\end{equation}
		
	\end{enumerate}
	
\end{algorithm}
We propose Algorithm \ref{alg:RW_Sampling} for estimating the fraction $\truevalue$ in case 1. The intuition behind Algorithm \ref{alg:RW_Sampling} stems from the fact that the stationary distribution of a random walk on an undirected graph (which is connected and non-bipartite) is the uniform distribution over the set of neighbors~\cite{aldous2002reversible}. Therefore, Algorithm \ref{alg:RW_Sampling} obtains a set $S$ of ${\samplingbudget}$ neighbors independently from the graph $G = (V,E)$ for sufficiently large~$N$ (i.e. one sample from each of the $\samplingbudget$ independent random walks) in the step~2. Then, the response $q(s)$ of each sampled individual $s\in S$ for the NEP query is used to compute the estimate $\estRW$ in step~3. According to the friendship paradox (Theorem \ref{th:friendship_paradox_Feld}), NEP with random neighbors is equivalent to using more node labels (than NEP with random nodes) due to the fact that random neighbors have more neighbors than random nodes on average. Hence, it is intuitive that the variance of this method should be smaller compared to the naive NEP (with uniformly sampled nodes) and intent polling method. In Sec. \ref{sec:analysis}, we verify this claim theoretically and, explore the properties of the underlying network for the estimate $\estRW$ to have a smaller MSE compared to the intent polling method.

\subsection{Case 2 - Sampling a Random Friend of a Random Individual}
\label{subsec:FN_method}

In case 1 (Sec.\ \ref{subsec:RW_method}), we assumed that it is possible to crawl the unknown graph using random walks. Instead, in case 2, we assume that a set of uniform samples $S = \{s_1,\dots, s_{\samplingbudget}\}$ from the set of nodes $V$ can be obtained and, each sampled individual $s_i \in S$ has the ability to answer the question \textit{"What is your (random) friend's estimate of the fraction of individuals with label 1?"}. 

\vspace{0.25cm}
{\noindent {\bf A motivating example} for case 2 is the situation where random individuals are requested to answer survey questions for an incentive. In such cases, the pollster usually does not have any information about the structural connectivity of the queried individuals and, will only be able to obtain their answer for a question. 
	
	For this case, we propose Algorithm \ref{alg:Friend_of_Node_Sampling} to obtain an estimate of the fraction $\truevalue$ of individuals with label $1$.
\begin{algorithm}
	\caption{NEP with Random Friend Sampling}
	\label{alg:Friend_of_Node_Sampling}
	\DontPrintSemicolon 
	\KwIn{${\samplingbudget}$ number of uniform samples $S = \{s_1, s_2,\dots, s_{\samplingbudget}\} \subset V$.}
	\KwOut{$\estFN$ which is the estimate of the fraction $\truevalue$ of nodes with label $1$.} 
	
	\vspace{0.3cm}
	
	\begin{enumerate}
		\item Ask each $s_i \in S$ to provide $q(u_i)$ (defined in (\ref{eq:response_to_query})) for some randomly chosen neighbor $u_i \in \mathcal{N}(s_i)$.

		\item Compute the estimate,
		\begin{equation}
		\estFN =  \frac{\sum_{i = 1}^{\samplingbudget} q(u_i)}{{\samplingbudget}}. \nonumber
		\end{equation}
		
	\end{enumerate}
	
\end{algorithm}

In Algorithm \ref{alg:Friend_of_Node_Sampling}, each uniformly sampled individual~${s_i \in S}$ answers the question \textit{``What is your (random) friend's estimate of the fraction of individuals with label 1?"} by providing $q(u_i)$ for a randomly chosen neighbor $u_i \in \mathcal{N}(s_i)$. The reasoning behind this method stems from Theorem~\ref{th:fosd_X_Z} which states that, a random friend of a randomly chosen individual has more friends than a randomly chosen individual on average\footnote{This does not follow from the original version of friendship paradox (Theorem \ref{th:friendship_paradox_Feld}) since the random friend is not a uniformly chosen neighbor from the set of all $2\vert E \vert$ neighbors. Instead, the response is obtained from a random neighbor of a uniformly sampled node. }. Therefore, this method should result in a smaller variance compared to naive NEP~(\ref{eq:naive_NEP}) and intent polling~(\ref{eq:intent_polling_estimate}).  

\begin{remark}
	\normalfont
One can think of Algorithm~\ref{alg:Friend_of_Node_Sampling} as a special case of Algorithm~\ref{alg:RW_Sampling} with the random walk length set to $N=1$. By the same argument, the naive NEP algorithm then corresponds to a random walk with length $N=0$ for the purpose of comparing the three NEP algorithms. The length of the random walk is used in Sec.~\ref{sec:discussion} to discuss how friendship paradox based NEP methods achieve a bias-variance trade-off. We refer the interested readers to~\cite{kramer2016multistep} which also explores the friendship paradox using random walk length. 
\end{remark}

\section{Statistical Analysis of the Estimates Obtained via Algorithm \ref{alg:RW_Sampling} and Algorithm \ref{alg:Friend_of_Node_Sampling}}
\label{sec:analysis}

Algorithm~\ref{alg:RW_Sampling} and Algorithm~\ref{alg:Friend_of_Node_Sampling} presented in Sec.\ \ref{sec:NEP_with_FP} query random friends (denoted by $Y$ in Theorem~\ref{th:friendship_paradox_Feld}) and random friends of random nodes (denoted by $Z$ in Theorem~\ref{th:fosd_X_Z})  respectively, exploiting the friendship paradox. In this context, the aim of this section is to analyze the bias, variance and the mean-squared error (MSE)\footnote{The mean-squared error~(MSE) of estimate $T$ of a parameter $\truevalue$ is
	\begin{align}
	\mse\{T\} &= \mathbb{E}\{(T - \bar{f})^2\}= \bias\{T\}^2 + \var\{T\} \label{eq:MSE}.
	\end{align}} of the estimates obtained using these proposed algorithms to show that they outperform alternative methods (intent polling and naive NEP without friendship paradox). More specifically,
\begin{compactenum}
\item Theorem~\ref{th:NEP_with_vs_without_FP} motivates the use of friendship paradox based NEP algorithms (compared to the naive NEP with uniformly sampled nodes) by considering the case where the label of each node is assigned by an independent and identically distributed coin toss.

\item Theorem~\ref{th:bias_var_Trw} relates bias and variance of the estimate $\estRW$ obtained using Algorithm \ref{alg:RW_Sampling} to network properties such as degree label correlation and absence of \mbox{bottlenecks}. Then, Corollary~\ref{cor:b_range} gives sufficient conditions on the sampling budget $\samplingbudget$ for which the Algorithm~\ref{alg:RW_Sampling} has a smaller MSE compared to intent polling.

\item Theorem~\ref{th:bias_var_Tun} characterizes the bias and variance of the naive NEP (with uniformly sampled nodes and hence, not exploiting friendship paradox) and, Corollary~\ref{cor:var_upperbounds_comparisn_Tun_Trw} compares the worst case performance of friendship paradox based NEP (Algorithm~\ref{alg:RW_Sampling}) with naive NEP to highlight how friendship paradox based sampling results in a reduced variance. 
\item Theorem~\ref{th:bias_var_Tfn} characterizes the bias and variance of the estimate $\estFN$ obtained using Algorithm~\ref{alg:Friend_of_Node_Sampling} and relates them to properties of the underlying network. 
 \end{compactenum}

\subsection{Independent and Identically Distributed Labels}
Consider graph $G = (V, E)$ where each node $v\in V$ has a binary label $f(v) \in \{0,1\}$ that is a Bernoulli random variable which is independent of and identically distributed to other labels. The following result shows how friendship paradox based sampling~(Algorithm~\ref{alg:RW_Sampling} and Algorithm~\ref{alg:Friend_of_Node_Sampling}) results in reduced variance NEP estimates.

\begin{theorem}
	\label{th:NEP_with_vs_without_FP}
	Let the set of labels $\{f(v): v \in V \}$ be independent and identically distributed~(iid) Bernoulli random variables. Then, 
	\begin{align}
		\mse\{\estFN\} &\leq \mse\{\estUN\}\\
		\mse\{\estRW\} &\leq \mse\{\estUN\}
	\end{align}
	where, $\mse$ denotes mean square error defined in (\ref{eq:MSE}), \newline $\estUN$ is the naive NEP estimate~(\ref{eq:naive_NEP}), 
	\newline $\estRW$ is the estimate obtained using Algorithm~\ref{alg:RW_Sampling}, 
	\newline $\estFN$ is the estimate obtained using Algorithm~\ref{alg:Friend_of_Node_Sampling}.
\end{theorem}
\begin{proof}
	By definition,
	\begin{align}
	\mathbb{E}\{\estRW\} &=\mathbb{E}\{q(Y)\} = \mathbb{E}\bigg\{\frac{\sum_{u\in \mathcal{N}(Y)}f(u)}{d(Y)}\bigg\} \nonumber\\
	\mathbb{E}\{\estFN\} &=\mathbb{E}\{q(Z)\} = \mathbb{E}\bigg\{\frac{\sum_{u\in \mathcal{N}(Z)}f(u)}{d(Z)}\bigg\}\nonumber\\
	\mathbb{E}\{\estUN\} &=\mathbb{E}\{q(X)\} = \mathbb{E}\bigg\{\frac{\sum_{u\in \mathcal{N}(X)}f(u)}{d(X)}\bigg\}.\nonumber
	\end{align}
	Consider 	$\mathbb{E}\{\estRW\}$.
	\begin{align}
	\mathbb{E}\{\estRW\} &= \mathbb{E}\bigg\{\frac{\sum_{u\in \mathcal{N}(Y)}f(u)}{d(Y)}\bigg\} \nonumber \\
	&=	\mathbb{E}\bigg\{	\mathbb{E}\bigg\{\frac{\sum_{i = 1}^{k}L_i}{k}	\bigg\vert d(Y) = k\bigg\}	\bigg\} \nonumber
	\end{align}
	where, $L_i, \; i = 1,\dots, k$ are the iid labels of the neighbors of $Y$. Since the labels $L_i$ are iid, the inner expectation becomes~$\mathbb{E}\{f(X)\}$. Therefore, 	
	\begin{align}
	\mathbb{E}\{\estRW\} = \mathbb{E}\{f(X)\} = \truevalue. \nonumber
	\end{align} 
	Following similar arguments, we also get,
	\begin{align}
	\mathbb{E}\{\estFN\} =\mathbb{E}\{\estUN\}=\mathbb{E}\{f(X)\} = \truevalue. \nonumber
	\end{align} 
	Therefore, the estimates are unbiased when the labels are iid. 
	
	Next, consider the variances of the estimate $\estRW$. Since all ${\samplingbudget}$ samples are independent,
	\begin{align}
	\var\{\estRW\} &= \frac{1}{{\samplingbudget}}\var\bigg\{\frac{\sum_{u\in \mathcal{N}(Y)}f(u)}{d(Y)} \bigg\}\nonumber
	\end{align}
	By applying the law of total variance, we get,
	\begin{align}
	\var\{\estRW\} &= \frac{1}{{\samplingbudget}}\bigg[	\var\bigg\{	\mathbb{E}\bigg\{\frac{\sum_{u\in \mathcal{N}(Y)}f(u)}{d(Y)}	\bigg\vert d(Y)\bigg\}	\bigg\}\,\,	+	\nonumber\\
	&\hspace{1cm}\mathbb{E}\bigg\{	\var\bigg\{\frac{\sum_{u\in \mathcal{N}(Y)}f(u)}{d(Y)}	\bigg\vert d(Y)\bigg\}	\bigg\} \bigg]\nonumber\\
	&=\frac{\sigma_f^2}{{\samplingbudget}}	\mathbb{E}\bigg\{\frac{1}{d(Y)}\bigg\}\nonumber \, \text{(since the labels are iid)}\nonumber
	\end{align}
	where, $\sigma_f^2$ denotes the variance of iid labels~i.e.~${\sigma_f^2 = \var\{f(X)\}}$. 
	Following similar steps, we obtain,
	\begin{align*}
	\var\{\estFN\} =\frac{\sigma_f^2}{{\samplingbudget}}	\mathbb{E}\bigg\{\frac{1}{d(Z)}\bigg\},\,\var\{\estUN\} =\frac{\sigma_f^2}{{\samplingbudget}}	\mathbb{E}\bigg\{\frac{1}{d(X)}\bigg\}.
	\end{align*}
	Then, the result follows by noting that 
	\begin{align}
	\frac{1}{d(X)} \geq_{fosd} \frac{1}{d(Y)}, \quad 
	\frac{1}{d(X)} \geq_{fosd} \frac{1}{d(Z)} \label{eq:inv_fosd_inv}
	\end{align} where, $\geq_{fosd}$ denotes the first order stochastic dominance defined in Footnote \ref{fn:fosd} in Sec. \ref{sec:fp}. Eq.~(\ref{eq:inv_fosd_inv}) follows immediately from Theorem \ref{th:friendship_paradox_Feld} and Theorem \ref{th:fosd_X_Z} (note that $d(\cdot)$ is strictly positive for connected graphs).
\end{proof}

Theorem \ref{th:NEP_with_vs_without_FP} shows that friendship paradox based NEP methods~(Algorithm~\ref{alg:RW_Sampling} and Algorithm~\ref{alg:Friend_of_Node_Sampling}) have smaller MSE compared to naive NEP~(\ref{eq:naive_NEP}) when the node labels are iid Bernoulli random variables. A natural question is ``How do friendship paradox based NEP methods perform when the node labels are assigned from an arbitrary joint distribution?''. We consider this next.

\subsection{Arbitrarily Assigned Node Labels}
In the remainder of this section, we assume that node labels $\{f(v): v\in V\}$ are already assigned from an arbitrary joint distribution or deterministically specified. 

We first characterize the bias $\bias\{\estRW\}$ and the variance $\var\{\estRW\}$ of the estimate $\estRW$ obtained via Algorithm \ref{alg:RW_Sampling} as the random walk length $N$ goes to infinity. Define the $|V| \times |V|$ dimensional diagonal matrix $D$ and the normalized adjacency matrix $\mathcal{A}$ as,
\begin{align}
D(v, v) = d(v), \quad \mathcal{A} = D^{-\frac{1}{2}}AD^{-\frac{1}{2}}. \label{eq:normalized_adj_matrix}
\end{align}
Let~$||Q||$ denote the spectral norm of a matrix~$Q$~(recall that the spectral norm is the maximum singular value). 
\begin{theorem}
	\label{th:bias_var_Trw} Let $G = (V, E)$ be a connected, non-bipartite graph. Then, as the random walk length $N$ tends to infinity, the bias~$\bias\{\estRW\}$ and the variance~$\var \{\estRW\}$ of the estimate~$\estRW$, obtained via Algorithm~\ref{alg:RW_Sampling} are given by,
	\begin{align}
	\begin{split} 	\label{eq:bias_T2}
	\bias(\estRW) &= \mathbb{E} \{f(Y)\} - \mathbb{E} \{f(X)\} \\
	&= \frac{\cov\{f(X),d(X)\}}{\mathbb{E}\{d(X)\}}
	\end{split}
	\end{align}
	\begin{align}
	\begin{split} 	\label{eq:var_ub_Trw}
	\hspace{-0.233cm}\var\{\estRW\} &= \frac{1}{\samplingbudget M}f^TD^{\frac{1}{2}}\bigg(   \mathcal{A}^2 - \frac{1}{M}D^{\frac{1}{2}}\mathds{1}\mathds{1}^TD^{\frac{1}{2}} \bigg)D^{\frac{1}{2}}f\\ 
	&\leq \frac{1}{{\samplingbudget}}\lambda_2^2 \mathbb{E}\{f(Y)\}
	\end{split}
	\end{align}
	where, $X$ is a random node, $Y$ is a random friend, $M$ is the total number of friends, $\lambda_2$ is the second largest singular value of the normalized adjacency matrix $\mathcal{A}$~(defined in~(\ref{eq:normalized_adj_matrix})) and $f$ is a column vector with label $f(v)\in \{0,1\}$ of node $v$ at $v^{th}$ element. 
\end{theorem}
\begin{proof}
If $G=(V,E)$ is a connected, non-bipartite graph, then the stationary distribution of a random walk on $G$ samples each ${v \in V}$ with a probability proportional to the degree $d(v)$ of $v$~(page 298, \cite{durrett2010_probability}). Equivalently, sampling from the stationary distribution of a random walk on a finite connected, non-bipartite graph is equivalent to sampling friendships $(U,Y) \in E$ uniformly. Therefore, 
\begin{align}
&\bias(\estRW) = \mathbb{E}\{\estRW\} - \truevalue = \mathbb{E}\{q(U)\} - \truevalue\nonumber\\
&\hspace{0.5cm}= \mathbb{E}\{f(Y)\} - \mathbb{E}\{f(X)\}\nonumber\\
&\hspace{0.5cm}= \sum_{v \in V}f(v)\frac{d(v)}{\sum_{v \in V}d(v)} - \frac{\sum_{v \in V}f(v)}{\vert V \vert}\nonumber\\
&\hspace{0.5cm} = \frac{\mathbb{E}\{f(X)d(X)\} - \mathbb{E}\{f(X)\}\mathbb{E}\{d(X)\} \nonumber }{\mathbb{E}\{d(X)\}}\\
&\hspace{0.5cm} = \frac{\cov\{f(X),d(X)\}}{\mathbb{E}\{d(X)\}}\nonumber
\end{align}

To obtain the variance of $q(Y)$, let $e_v$ denote the $n\times1$ dimensional unit vector with $1$ at the $v^{th}$ element and zeros elsewhere. Then, $q(v) = e_v^T	D^{-1}Af$. Hence, 
\begin{align}
\mathbb{E}\{q(Y)\} &= \sum_{v \in V} \frac{d(v)}{M}e_v^T	D^{-1}Af = \frac{1}{M}\mathds{1}^TD	D^{-1}Af \nonumber\\
&=\frac{1}{M}\mathds{1}^TAf = \frac{1}{M}\mathds{1}^TDf
\end{align}
\begin{align}
\mathbb{E}\{q^2(Y)\} &= \sum_{v \in V} \frac{d(v)}{M} f^TAD^{-1}e_v e_v^T	D^{-1}Af  \nonumber\\
&=\frac{1}{M}f^TAD^{-1}Af.
\end{align}
Therefore,
\begin{align}
\var\{q(Y)\} &= \mathbb{E}\{q^2(Y)\} - \mathbb{E}\{q(Y)\}^2\nonumber\\
&\hspace{-1.5cm}= \frac{1}{M}f^TAD^{-1}Af - \frac{1}{M^2}f^TD\mathds{1}\mathds{1}^TDf\nonumber\\
&\hspace{-1.5cm}=\frac{1}{M}f^TD^{\frac{1}{2}} \bigg(\Big(D^{-\frac{1}{2}}AD^{-\frac{1}{2}}\Big)^2 - \Big(\frac{D^{\frac{1}{2}}\mathds{1}}{\sqrt{M}}\Big)\Big(\frac{\mathds{1}^TD^{\frac{1}{2}} }{\sqrt{M}}\Big)\bigg)D^{\frac{1}{2}}f\nonumber\\
&\hspace{-1.5cm}=\frac{1}{M}f^TD^{\frac{1}{2}} \bigg(\mathcal{A}^2 - \Big(\frac{D^{\frac{1}{2}}\mathds{1}}{\sqrt{M}}\Big)\Big(\frac{\mathds{1}^TD^{\frac{1}{2}} }{\sqrt{M}}\Big)\bigg)D^{\frac{1}{2}}f,\nonumber
\end{align}
where~$\mathcal{A}$ denotes the normalized adjacency matrix defined in~(\ref{eq:normalized_adj_matrix}). Note that $\frac{D^{\frac{1}{2}}\mathds{1}}{\sqrt{M}}$ is the eigenvector corresponding to the largest eigenvalue $1$ of $\mathcal{A}^2$. Therefore, we get
\begin{align}
& \bigg|\frac{1}{M}f^TD^{\frac{1}{2}}\bigg(\mathcal{A}^2 - \Big(\frac{D^{\frac{1}{2}}\mathds{1}}{\sqrt{M}}\Big)\Big(\frac{\mathds{1}^TD^{\frac{1}{2}} }{\sqrt{M}}\Big)\bigg)D^{\frac{1}{2}}f\bigg| \nonumber\\
&\hspace{0.5cm}\leq \bigg|\bigg|\frac{D^{\frac{1}{2}}f}{\sqrt{M}}\bigg|\bigg|\times\bigg|\bigg|\bigg(\mathcal{A}^2 - \Big(\frac{D^{\frac{1}{2}}\mathds{1}}{\sqrt{M}}\Big)\Big(\frac{\mathds{1}^TD^{\frac{1}{2}} }{\sqrt{M}}\Big)\bigg) \frac{D^{\frac{1}{2}}f}{\sqrt{M}} \bigg|\bigg|\nonumber\\
&\hspace{1cm} \text{(by Cauchy-Schwarz inequality)}\nonumber\\
&\hspace{0.5cm}\leq \bigg|\bigg|\bigg(\mathcal{A}^2 - \Big(\frac{D^{\frac{1}{2}}\mathds{1}}{\sqrt{M}}\Big)\Big(\frac{\mathds{1}^TD^{\frac{1}{2}} }{\sqrt{M}}\Big)\bigg) \bigg|\bigg|\times \bigg|\bigg|\frac{D^{\frac{1}{2}}f}{\sqrt{M}}\bigg|\bigg|^2 \nonumber\\
&\hspace{1cm} \text{(where, $||Q||$ denotes operator norm of a matrix $Q$)}\nonumber\\
&\hspace{0.5cm}= \lambda_2^2 \mathbb{E}\{f(Y)\}\nonumber
\end{align}
and~(\ref{eq:var_ub_Trw}) follows. 
\end{proof}

Theorem~\ref{th:bias_var_Trw} gives insight into the network properties that affect the performance of the Algorithm~\ref{alg:RW_Sampling}. Eq.~(\ref{eq:bias_T2}) states that, the bias of the estimate $\estRW$ is proportional to the covariance between the degree $d(X)$ and the label $f(X)$ of a random node $X$. Theorem~\ref{th:bias_var_Trw} also shows that the variance of the estimate $\estRW$ is bounded above by a function of the second largest singular value $\lambda_2$ of the normalized adjacency matrix $\mathcal{A}$ and the expected label value of a random friend $Y$. Hence, a smaller $\lambda_2$ which indicates that the network has a good expansion\footnote{A network is considered to have ``good expansion" if every subset $S$ of
	nodes ($S \leq 50\%$ of the nodes) has a neighborhood that
	is larger than some ``expansion factor'' multiplied by the
	number of nodes in~$S$. Hence, a good expansion factor indicates that that there are no bottlenecks i.e. there is no small set of edges whose removal will fragment the network into two large connected components~	\cite{estrada2006network}.} (i.e. absence of bottlenecks)~\cite{estrada2006network} will result in a smaller variance in the estimate $\estRW$. 

The following corollary gives a sufficient condition for the estimate $\estRW$ to be more statistically efficient (i.e.~smaller MSE) compared to the classical intent polling method. Recall that the sampling budget $\samplingbudget$ denotes the number of nodes queried by the pollster.
\begin{corollary}
	\label{cor:b_range}
	If the sampling budget $\samplingbudget$ satisfies
	\begin{align}
	\samplingbudget \leq \frac{\big(\var\{f(X)\} - \lambda_2^2\mathbb{E}\{f(Y)\}\big)\mathbb{E}\{d(X)\}^2}{\cov\{f(X)d(X)\}^2},
	\end{align} then the estimate $\estRW$ obtained from Algorithm~\ref{alg:RW_Sampling} has a smaller MSE compared to the intent polling estimate $\estIP$ in (\ref{eq:intent_polling_estimate}), i.e.~$\mse\{\estRW\} \leq \mse\{\estIP\}$.
\end{corollary}
\begin{proof}
From (\ref{eq:bias_T2}) and (\ref{eq:var_ub_Trw}) we get, 
	\begin{align}
		\mse\{\estRW\} &= \bias\{\estRW\}^2 + \var\{\estRW\}	\nonumber\\
		&\hspace{-1cm}\leq \bigg(	\frac{\cov\{f(X),d(X)\}}{\mathbb{E}\{d(X)\}} \bigg)^2 + \frac{\lambda_2^2 \mathbb{E}\{f(Y)\}}{{\samplingbudget}}.\label{eq:mse_Trw_ub}\\
		&\hspace{-2cm}\text{Also,}	\nonumber\\
		\mse\{\estIP\} &=\var\{\estIP\} =\frac{\var\{f(X)\}}{\samplingbudget}. \label{eq:mse_IP}
	\end{align}
Hence, the result follows from (\ref{eq:mse_Trw_ub}) and (\ref{eq:mse_IP}).
\end{proof}
Corollary~\ref{cor:b_range} indicates that a smaller degree-label correlation and the absence of bottlenecks result in the estimate~$\estRW$ outperforming intent polling~(\ref{eq:intent_polling_estimate}) for a larger range of sampling budgets~${\samplingbudget}$. This is because smaller label-degree correlation and the absence of bottlenecks make the bias and variance of $\estRW$ smaller according to Theorem~\ref{th:bias_var_Trw} and therefore, make the MSE of $\estRW$ smaller.

Next, we characterize bias and variance of the naive NEP estimate~$\estUN$~(defined in (\ref{eq:naive_NEP})), thereby allowing us to compare it with friendship paradox based NEP methods (Algorithm~\ref{alg:RW_Sampling} and Algorithm~\ref{alg:Friend_of_Node_Sampling}). 
\begin{theorem}
	\label{th:bias_var_Tun}
	The bias $\bias\{\estUN\}$ and the variance $\var \{\estUN\}$ of the naive NEP estimate $\estUN$~(defined in~(\ref{eq:naive_NEP})) are given by,
	\begin{align}
	\label{eq:bias_Tfn_eq1}
	\bias(\estUN) &= \mathbb{E} \{f(Z)\} - \mathbb{E} \{f(X)\} \\
	\var\{\estUN\} &= \frac{1}{\samplingbudget n}f^TD^{\frac{1}{2}}   \mathcal{A}D^{-\frac{1}{2}}\bigg(I - \frac{\mathds{1}\mathds{1}^T}{n} \bigg)D^{-\frac{1}{2}}\mathcal{A}D^{\frac{1}{2}}f \nonumber\\
	&\leq \frac{1}{\samplingbudget}\frac{\mathbb{E}\{f(Y)\}\mathbb{E}\{d(X)\}}{d_{min}} \label{eq:var_ub_Tun}
	\end{align}
where, $n$ is the total number of nodes, $X$ is a random node, $Y$ is a random friend, $Z$ is a random friend of a random node,  $\mathcal{A}$ is the normalized adjacency matrix defined in~(\ref{eq:normalized_adj_matrix}) and $f$ is a column vector with label $f(v)\in \{0,1\}$ of node $v$ at $v^{th}$ element. 
\end{theorem}
\begin{proof}
{Note that,}
\begin{align}
\mathbb{E}\{q(X)\} &= \mathbb{E}\bigg\{\frac{\sum_{u\in \mathcal{N}(X)}f(u)}{d(X)}\bigg\}\nonumber \\
&= \mathbb{E}\big\{\mathbb{E}\big\{ f(Z)\vert X\big\} \big\} = \mathbb{E}\{f(Z)\},\nonumber
\end{align}	
{from which, (\ref{eq:bias_Tfn_eq1}) follows.}

\noindent
Next, recall that $q(v) = e_v^T	D^{-1}Af$. 	Hence,
\begin{align}
\mathbb{E}\{q(X)\} &= \sum_{v \in V} \frac{1}{n}e_v^T	D^{-1}Af = \frac{1}{n}\mathds{1}^T	D^{-1}Af\;\text{and,} \nonumber\\
\mathbb{E}\{q^2(X)\} &= \sum_{v \in V} \frac{1}{n} f^TAD^{-1}e_v e_v^T	D^{-1}Af  \nonumber\\
&=\frac{1}{n}f^TAD^{-2}Af \nonumber
\end{align}
Therefore,
\begin{align}
\var\{q(X)\} &= \mathbb{E}\{q^2(X)\} - \mathbb{E}\{q(X)\}^2\nonumber\\
&\hspace{-1.5cm}= \frac{1}{n}f^TAD^{-2}Af - \frac{1}{n^2}f^TAD^{-1}\mathds{1}\mathds{1}^TD^{-1}Af\nonumber\\
&\hspace{-1.5cm}=\frac{1}{n}f^TD^{\frac{1}{2}} \bigg(\Big(D^{-\frac{1}{2}}AD^{-\frac{1}{2}}\Big)D^{-1}\Big(D^{-\frac{1}{2}}AD^{-\frac{1}{2}}\Big) \nonumber\\ 
&-\frac{1}{n}D^{-\frac{1}{2}}AD^{-1}\mathds{1}\mathds{1}^TD^{-1}AD^{-\frac{1}{2}}\bigg)D^{\frac{1}{2}}f\nonumber\\
&\hspace{-1.5cm}=\frac{1}{n}\big(f^TD^{\frac{1}{2}} \big)\big(\mathcal{A}D^{-\frac{1}{2}}\big) \bigg(I- \frac{\mathds{1}\mathds{1}^T}{n}\bigg)(D^{-\frac{1}{2}}\mathcal{A}\big) \big(D^{\frac{1}{2}}f\big)\nonumber
\end{align}
{By the sub-multiplicative property of matrix norms, }
\begin{align}
&\bigg|\bigg|			\big(\mathcal{A}D^{-\frac{1}{2}}\big) \bigg(I- \frac{\mathds{1}\mathds{1}^T}{n}\bigg)(D^{-\frac{1}{2}}\mathcal{A}\big)		\bigg|\bigg|\nonumber\\	
&\hspace{1.5cm} \leq ||\mathcal{A}||^2 ||D^{-\frac{1}{2}}||^2\bigg|\bigg| I- \frac{\mathds{1}\mathds{1}^T}{n} \bigg|\bigg|^2 = \frac{1}{d_{min}}	\label{eq:sub_multiplicative_matrix_norm}\\
&\hspace{0cm} \text{(where, $||Q||$ denotes operator norm of a matrix $Q$).}\nonumber
\end{align}
{Therefore, by applying Cauchy-Schwarz inequality and then using~(\ref{eq:sub_multiplicative_matrix_norm}), we get}
\begin{align}
\var\{q(X)\}  &\leq  \frac{||D^{\frac{1}{2}}f||^2}{n} \frac{1}{d_{min}}= \mathbb{E}\{f(Y)\} \frac{\mathbb{E}\{d(X)\}}{d_{min}}, \nonumber
\end{align}
{and (\ref{eq:var_ub_Tun}) follows.}
\end{proof}

The following corollary is a consequence of Theorem~\ref{th:bias_var_Trw} and Theorem~\ref{th:bias_var_Tun}. It compares the worst case performances of friendship paradox based NEP estimate $\estRW$~(obtained via Algorithm~\ref{alg:RW_Sampling}) and naive NEP estimate~$\estUN$~(defined in (\ref{eq:naive_NEP})). The result shows how friendship paradox based sampling reduces variance of NEP methods. 
\begin{corollary}
	\label{cor:var_upperbounds_comparisn_Tun_Trw}
	The upper bound (\ref{eq:var_ub_Trw}) for the variance of the estimate $\estRW$ (from Algorithm~\ref{alg:RW_Sampling}) and the upper bound (\ref{eq:var_ub_Tun}) for the variance of the estimate $\estUN$  (naive NEP) satisfy,
	\begin{align}
\frac{1}{{\samplingbudget}}\lambda_2^2 \mathbb{E}\{f(Y)\}\leq 	\frac{1}{\samplingbudget}\frac{\mathbb{E}\{f(Y)\}\mathbb{E}\{d(X)\}}{d_{min}}.
	\end{align}
\end{corollary}
\begin{proof}
	The proof follows by the fact that
	${0\leq \lambda_2^2  < 1 \leq \frac{\mathbb{E}\{d(X)\}}{d_{min}}}$.
\end{proof}

Finally, we characterize  bias and variance of the estimate~$\estFN$ obtained via Algorithm~\ref{alg:Friend_of_Node_Sampling} which exploits the second version of the friendship paradox (Theorem~\ref{th:fosd_X_Z}).
\begin{theorem}
	\label{th:bias_var_Tfn}
	The bias $\bias\{\estFN\}$ and the variance $\var \{\estFN\}$ of the estimate $\estFN$, obtained via Algorithm~\ref{alg:Friend_of_Node_Sampling} satisfy,
	\begin{align}
	\begin{split} \label{eq:bias_ub_Tfn}
	\bias\{\estFN\}^2 &= \frac{1}{n}\mathds{1}^TD^{-\frac{1}{2}}\big(\mathcal{A}^2	- I\big)D^{\frac{1}{2}}f\\
	&\leq (\lambda_n^2 -1)^2 \mathbb{E}\{f(Y)\}\frac{\mathbb{E}\{d(X)\}}{\bar{d}_{hm}} 
	\end{split}\\
	\var\{\estFN\} &= \frac{1}{\samplingbudget n}f^T  {A}D_{hm}^{-\frac{1}{2}}\bigg(D^{-1} - \frac{D_{hm}^{-\frac{1}{2}}\mathds{1}\mathds{1}^TD_{hm}^{-\frac{1}{2}}}{n} \bigg)D_{hm}^{-\frac{1}{2}}Af 	\label{eq:var_Tfn}
	\end{align}
	where, $\lambda_n$ is the smallest singular value of the normalized adjacency matrix $\mathcal{A}$, $\bar{d}_{hm} = {\mathbb{E} \big\{\frac{1}{d(X)}\big\}}^{-1}$ is the harmonic mean degree of the graph and $D_{hm}$ is a diagonal matrix with harmonic mean of the neighbor degrees of node $v\in V$ at the $v^{th}$ element.
\end{theorem}
\begin{proof}
Note that  $\mathbb{P}\{Z = v\} = \frac{1}{n}e_v^TAD^{-1}\mathds{1}$ and recall that $q(v) =  e_v^T	D^{-1}Af$. 	Hence,	
\begin{align}
\mathbb{E}\{q(Z)\} &= \sum_{v \in V}\mathbb{P}\{Z=v\} e_v^T	D^{-1}Af			\nonumber\\
&=   \sum_{v \in V}\frac{1}{n}(\mathds{1}^TD^{-1}A e_v	)(e_v^T	D^{-1}Af) 	\nonumber\\
&= \frac{1}{n}\mathds{1}^TD^{-1}A D^{-1}Af 		\label{eq:Exp_qz}
\end{align}
Following similar steps to the above, we get,
\begin{align}
\mathbb{E}\{q^2(Z)\} &= \sum_{v \in V}\mathbb{P}\{Z=v\}f^T	A D^{-1}e_ve_v^T	D^{-1}Af 	\nonumber\\
&= \frac{1}{n}f^TAD^{-1}\Big(		\sum_{v \in V}e_ve_v^TAD^{-1}\mathds{1}e_v^T		\Big)D^{-1}Af \nonumber\\
&= \frac{1}{n}f^TAD_{hm}^{-1}{D^{-1}}Af \quad \text{where,} \label{eq:Exp_qz_squared}
\end{align}
$D_{hm}$ is a diagonal matrix with harmonic mean of the neighbors of node $v\in V$ at $v^{th}$ diagonal element i.e.~${D_{hm}(v,v) = d(v)\Big(\sum_{u\in \mathcal{N}(v)}\frac{1}{d(u)}\Big)^{-1}}$.
Then, (\ref{eq:var_Tfn}) follows from (\ref{eq:Exp_qz}) and (\ref{eq:Exp_qz_squared}).

Next we prove (\ref{eq:bias_ub_Tfn}).
\begin{align}
\bias\{\estFN\} &= \mathbb{E}\{q(Z)\}  - \mathbb{E}\{f(X)\} 			\nonumber\\
&=\frac{1}{n}\mathds{1}^TD^{-1}A D^{-1}Af 	 - \frac{\mathds{1}^Tf}{n} \nonumber\\
&=\frac{\mathds{1}^TD^{-\frac{1}{2}}}{n} \Big(			\mathcal{A}^2 - I				\Big)D^{\frac{1}{2}}f \nonumber
\end{align}
Hence,
\begin{align}
|\bias\{\estFN\}| &\leq  \bigg|\bigg|	\frac{\mathds{1}^TD^{-\frac{1}{2}}}{n} \Big(			\mathcal{A}^2 - I				\Big)D^{\frac{1}{2}}f \bigg|\bigg| \nonumber	\\
&= \frac{1}{n}(\lambda_n^2-1)||D^{\frac{1}{2}}f||\times ||\mathds{1}^TD^{-\frac{1}{2}}|| \nonumber
\end{align}
{which implies,} 
\begin{align}
\bias\{\estFN\}^2 &\leq \frac{M}{n}(\lambda_n^2-1)^2\frac{||D^{\frac{1}{2}}f||^2}{M}\times \frac{\sum_{v \in V}\frac{1}{d(v)}}{n} \nonumber\\
&=(\lambda_n^2-1)^2\mathbb{E}\{f(Y)\}\times \frac{\mathbb{E}\{d(X)\}}{\bar{d}_{hm}} \nonumber
\end{align}
and (\ref{eq:bias_ub_Tfn}) follows. 
\end{proof}

Eq.~(\ref{eq:bias_ub_Tfn}) shows that the bias of the estimate $\estFN$ depends on the smallest singular value of the normalized adjacency matrix $\mathcal{A}$. This suggests that, the bias of the estimate $\estFN$ based on second version of friendship paradox depends on spectral properties of the network as opposed to the estimate $\estRW$~(obtained via Algorithm~\ref{alg:RW_Sampling}) based on the first version of the friendship paradox~(Theorem~\ref{th:friendship_paradox_Feld}). 

\vspace{0.25cm}
\noindent
{\bf Summary of Statistical Analysis:} The above results~(Theorem~\ref{th:NEP_with_vs_without_FP} to Theorem~\ref{th:bias_var_Tfn}) motivate the use of NEP with friendship paradox based sampling (Algorithm~\ref{alg:RW_Sampling} and Algorithm~\ref{alg:Friend_of_Node_Sampling}) compared to the intent polling and NEP without friendship paradox (i.e. naive NEP). Theorem~\ref{th:NEP_with_vs_without_FP} showed that the two friendship paradox based NEP algorithms have smaller MSE compared to the naive NEP method when labels are independently and identically distributed. Then, Theorem~\ref{th:bias_var_Trw} characterized the bias and variance of the estimate $\estRW$ obtained via Algorithm~\ref{alg:RW_Sampling} and Corollary~\ref{cor:b_range} illustrated that it has a smaller MSE compared to intent polling for small sampling budget ${\samplingbudget}$ values. Further, Theorem \ref{th:bias_var_Trw} also showed that the bias and variance of the estimate $\estRW$ are affected by the degree-label correlation and the expansion of the network respectively. Next, Theorem~\ref{th:bias_var_Tun} characterized the bias and variance of the naive NEP estimate $\estUN$ and Corollary~\ref{cor:var_upperbounds_comparisn_Tun_Trw} illustrated how NEP with friendship paradox outperforms naive NEP (without friendship paradox). Finally, Theorem~\ref{th:bias_var_Tfn} characterized the bias and variance of estimate $\estFN$ produced by the Algorithm~\ref{alg:Friend_of_Node_Sampling} based on the second version of the friendship paradox (Theorem~\ref{th:fosd_X_Z}). It shows that the bias of estimate $\estFN$ depends on the spectral properties of the network as opposed the estimate $\estRW$ based on the first version of the friendship paradox.

\section{Empirical and Simulation Results}
\label{sec:experiments}
The aim of this section is to evaluate Algorithm~\ref{alg:RW_Sampling} and Algorithm~\ref{alg:Friend_of_Node_Sampling} on five large scale real world social networks as well as synthetic network datasets in order to obtain insights that complement the analytical results presented in Sec.~\ref{sec:analysis}. More specifically,
\begin{compactenum}
\item Sec.~\ref{subsec:experiments_real_world_SNs} evaluates Algorithm~\ref{alg:RW_Sampling}, Algorithm~\ref{alg:Friend_of_Node_Sampling}, naive NEP and intent polling on five real world social networks with different degree-label correlation coefficients. 

\item Sec.~\ref{subsec:simulations} evaluates Algorithm~\ref{alg:RW_Sampling}, Algorithm~\ref{alg:Friend_of_Node_Sampling}, naive NEP and intent polling on networks that are obtained from two well known models: configuration model \cite{molloy1995critical} and  Erd\H{o}s-R\'{e}nyi~($G(n,p)$) model\cite{newman2002random}. 
\end{compactenum}
The key conclusions that can be drawn from these experiments and simulations, and how they relate to the analytical results, are then discussed in detail in Sec.~\ref{sec:discussion}. 

Before proceeding to present the results, we define three key variables that are widely used in social network analysis. 
\begin{compactenum}
	\item {\bf Degree distribution} $P(k)$ is the probability that a randomly chosen node has $k$ neighbors.
	\vspace{0.25cm}
	
	\item {\bf Neighbor degree correlation (assortativity) coefficient} is defined as,
	\begin{equation}
	\label{eq:deg_deg_corr}
	r_{kk} = \frac{1}{\sigma_q^2}\sum_	{k,k'}kk'\Big(e(k,k')  - q(k)q(k')\Big)
	\end{equation} where, $e(k,k')$ is the probability of nodes at the ends of a randomly chosen edge have degrees $k$ and $k'$ (joint degree distribution of neighbors), $q(k)$ is the probability that a random neighbor has $k$ neighbors (marginal distribution of $e(k,k')$) and $\sigma_q$ is the standard deviation with respect to $q$.
	\vspace{0.25cm}
	
	\item {\bf Degree-label correlation coefficient} is defined as,
	\begin{align}
		\rho_{kf} &= \frac{1}{\sigma_k\sigma_f}\sum_{k}k\Big(\mathbb{P}(f{(X)} = 1,d(X) = k) \nonumber\\
		&\hspace{3.0cm}- \mathbb{P}(f{(X)} = 1)P(k) \Big) 	\label{eq:deg_label_corr}
	\end{align}
	where, $\sigma_k$ and $\sigma_f$ are the standard deviations of the degree of a random node and the label of a random node respectively.
\end{compactenum}
A detailed discussion of these variables and their effects can be found in \cite{lerman2016}.

\subsection{Real World Networks}
\label{subsec:experiments_real_world_SNs}

{\bf Dataset Description: }The datasets used in this subsection are openly available from the Stanford Network Analysis Project (SNAP) \cite{snapnets}. Below, we describe each dataset briefly.
\begin{compactenum}
	\item \textit{Facebook Social Circles} \cite{leskovec2012learning}: This dataset consists of ``circles" (or ``friends lists") from Facebook that were collected using the Facebook App. Total number of nodes and edges in the network constructed from this dataset are $4039$ and $88234$ respectively. The neighbor degree correlation coefficient $r_{kk}$ (defined in (\ref{eq:deg_deg_corr})) of the network is $0.06$
	
	\item \textit{Co-authorship Network} \cite{leskovec2007graph}: This dataset contains the scientific collaborations between authors of papers submitted to General Relativity and Quantum Cosmology category in the Arxiv website. More specifically, an author $i$ co-authoring a paper with author $j$ will be represented by an undirected edge between the two nodes $i$ and $j$ in the network. Total number of nodes and edges in the network constructed from this dataset are $5242$ and $14496$ respectively. The neighbor degree correlation coefficient $r_{kk}$ (defined in (\ref{eq:deg_deg_corr})) of the network is $0.66$.
	
	\item \textit{Athlete Network} \cite{rozemberczki2018gemsec}: This dataset contains Facebook page networks of athletes. The nodes in the network represent the Facebook pages of athletes and the edges represent mutual likes among them. Total number of nodes and edges in the network constructed from this dataset are $13,866$ and $86,858$ respectively. The neighbor degree correlation coefficient $r_{kk}$ (defined in (\ref{eq:deg_deg_corr})) of the network is $-0.03$.
	
	\item \textit{Politician Network} \cite{rozemberczki2018gemsec}: This dataset contains Facebook page networks of politicians. The nodes in the network represent the Facebook pages of politicians and the edges represent mutual likes among them. Total number of nodes and edges in the network constructed from this dataset are $5908$ and $41729$ respectively. The neighbor degree correlation coefficient $r_{kk}$ (defined in (\ref{eq:deg_deg_corr})) of the network is $0.02$.
	
	\item \textit{Company Network} \cite{rozemberczki2018gemsec}: This dataset contains Facebook page networks of different companies. The nodes in the network represent the Facebook pages of companies and the edges represent mutual likes among them. Total number of nodes and edges in the network constructed from this dataset are $14,113$ and $52,310$ respectively. The neighbor degree correlation coefficient $r_{kk}$ (defined in (\ref{eq:deg_deg_corr})) of the network is $0.01$.
\end{compactenum}

\vspace{0.25cm}
\noindent
{\bf Label swapping procedure for modifying degree-label correlation:} Given a graph $G = (V,E)$, we first assign labels $f(v)$ to each node $v \in V$ with a fixed probability. Then, to set the degree-label correlation coefficient defined in (\ref{eq:deg_label_corr}) to a desired value, we utilize the label swapping procedure followed in~\cite{lerman2016}: a node $v_0$ with a label $f(v_0) = 0$ and a node $v_1$ with a label $f(v_1) = 1$ are selected randomly and their labels are swapped if $d(v_0) < d(v_1)$ (respectively, $d(v_0) > d(v_1)$) to increase (respectively, decrease) the degree-label correlation coefficient $\rho_{kf}$ to the desired value (or until it no longer changes).  We consider $\rho_{kf} = -0.1, 0, 0.1$ in our experiments to study the effect of negative and positive degree-label correlations on the accuracy of the polling algorithms.

\vspace{0.25cm}
\noindent
{\bf Empirical Results: } The MSE and variance of the four polling methods (Algorithm~\ref{alg:RW_Sampling}, Algorithm~\ref{alg:Friend_of_Node_Sampling}, intent polling and naive NEP) were estimated using Monte-Carlo simulation over $600$ independent iterations for each value of the sampling budget ${\samplingbudget}$ from~$1$ to approximately~$1\%$ of the total number of nodes in the network. The results are displayed in~Fig.~\ref{fig:empirical_results}. The conclusions and insights that can be drawn from these empirical results and how they relate to the analytical results are discussed in Sec.~\ref{sec:discussion}.

\begin{figure*}[]
	\centering
	\begin{subfigure}[!h]{0.3\textwidth}
		\centering
		\includegraphics[width=2.6in]{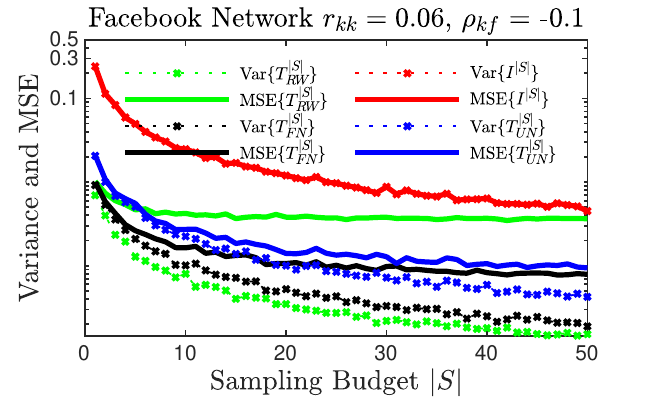}
		\caption{}
	\end{subfigure}\hfill
	\begin{subfigure}[!h]{0.3\textwidth}
		\centering
		\includegraphics[width=2.6in]{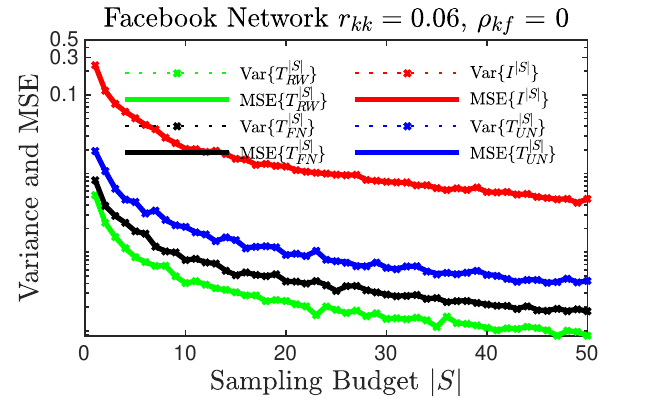}
		\caption{}
	\end{subfigure}\hfill 
	\begin{subfigure}[!h]{0.3\textwidth}
		\centering
		\includegraphics[width=2.6in] {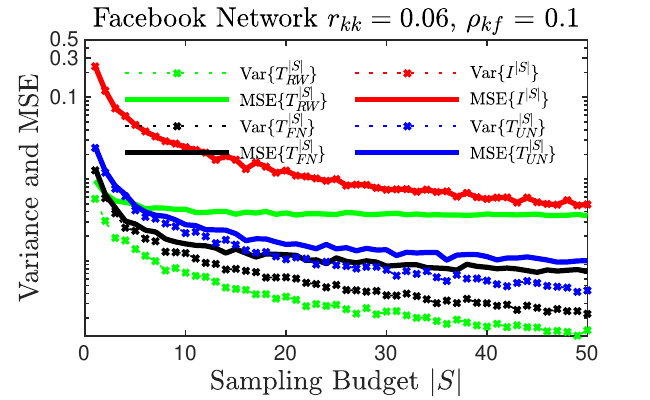}
		\caption{}
	\end{subfigure}

	\begin{subfigure}[!h]{0.3\textwidth}
		\centering
		\includegraphics[width=2.6in]{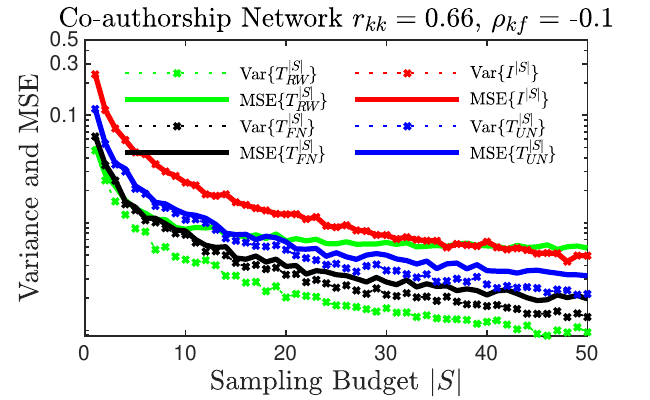}
		\caption{}
	\end{subfigure}\hfill
	\begin{subfigure}[!h]{0.3\textwidth}
		\centering
		\includegraphics[width=2.6in]{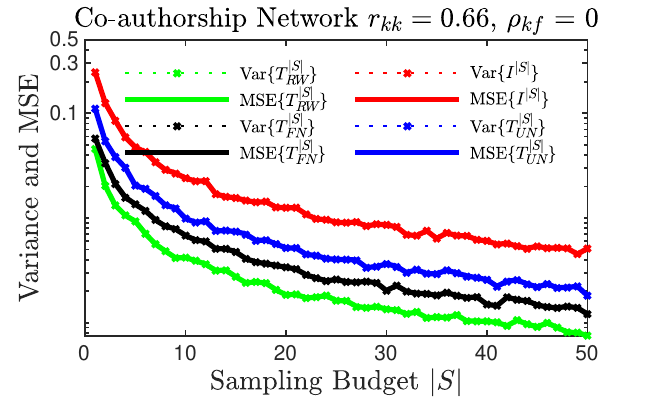}
		\caption{}
	\end{subfigure}\hfill 
	\begin{subfigure}[!h]{0.3\textwidth}
		\centering
		\includegraphics[width=2.6in] {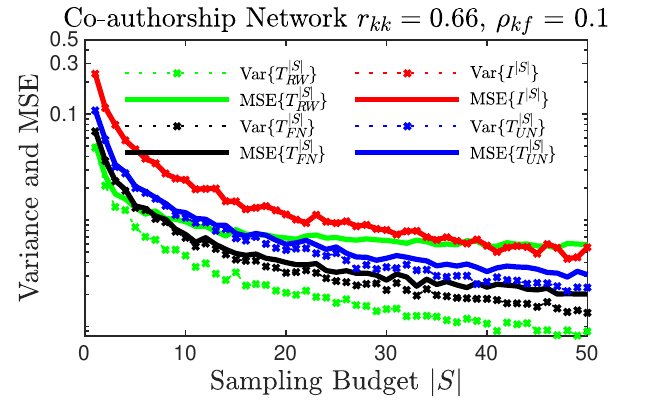}
		\caption{}
	\end{subfigure}

	\begin{subfigure}[!h]{0.3\textwidth}
		\centering
		\includegraphics[width=2.6in]{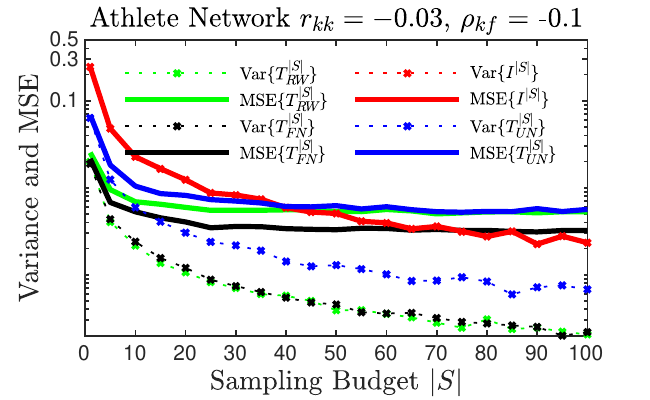}
		\caption{}
	\end{subfigure}\hfill
	\begin{subfigure}[!h]{0.3\textwidth}
		\centering
		\includegraphics[width=2.6in]{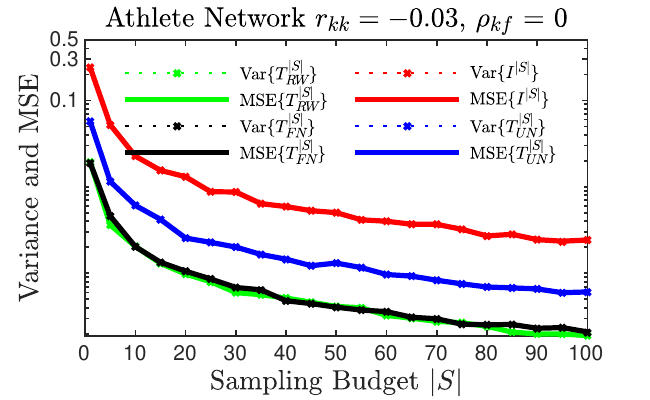}
		\caption{}
	\end{subfigure}\hfill 
	\begin{subfigure}[!h]{0.3\textwidth}
		\centering
		\includegraphics[width=2.6in] {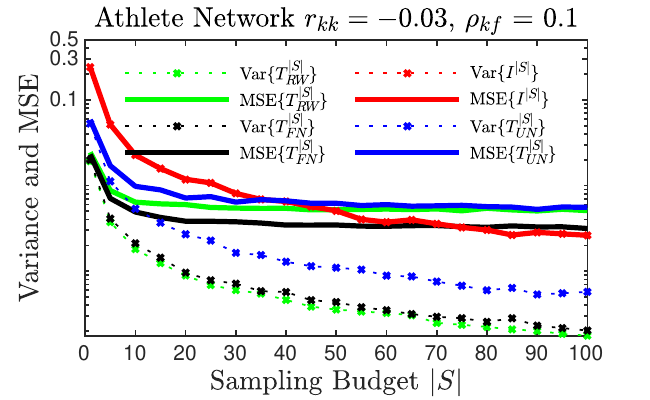}
		\caption{}
	\end{subfigure}

	\begin{subfigure}[!h]{0.3\textwidth}
		\centering
		\includegraphics[width=2.6in]{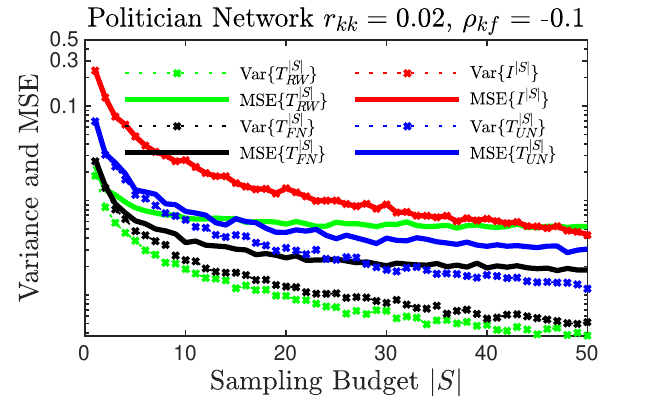}
		\caption{}
	\end{subfigure}\hfill
	\begin{subfigure}[!h]{0.3\textwidth}
		\centering
		\includegraphics[width=2.6in]{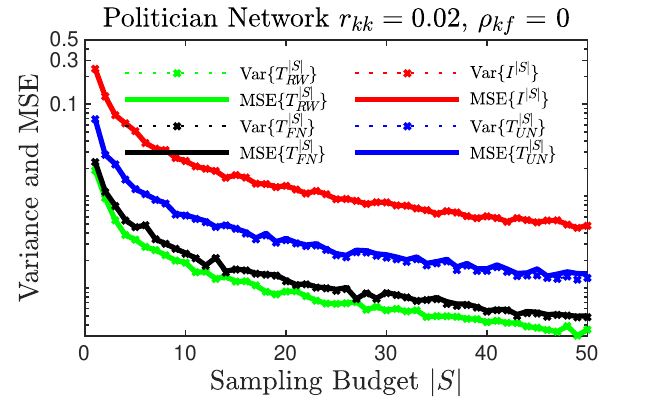}
		\caption{}
	\end{subfigure}\hfill 
	\begin{subfigure}[!h]{0.3\textwidth}
		\centering
		\includegraphics[width=2.6in] {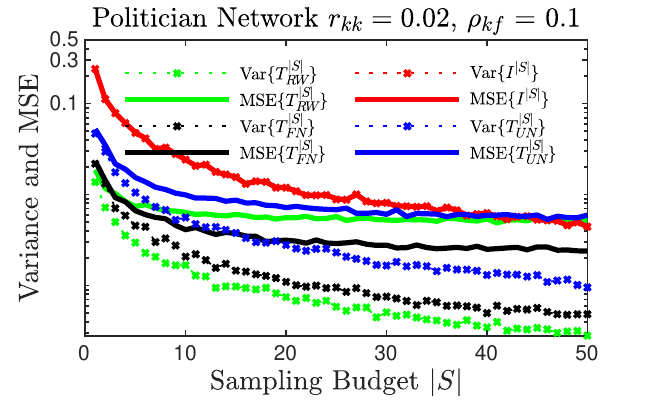}
		\caption{}
	\end{subfigure} 
%
\ContinuedFloat
	\begin{subfigure}[!h]{0.3\textwidth}
		\centering
		\includegraphics[width=2.6in]{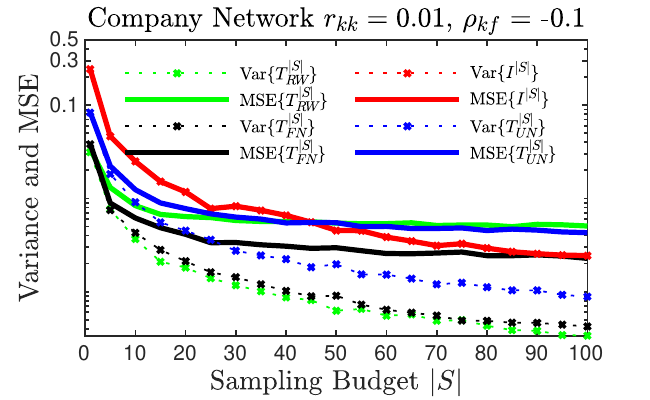}
		\caption{}
	\end{subfigure}\hfill
	\begin{subfigure}[!h]{0.3\textwidth}
		\centering
		\includegraphics[width=2.6in]{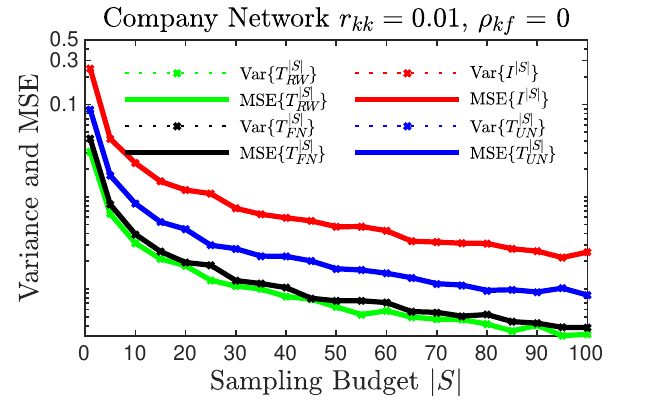}
		\caption{}
	\end{subfigure}\hfill 
	\begin{subfigure}[!h]{0.3\textwidth}
		\centering
		\includegraphics[width=2.6in] {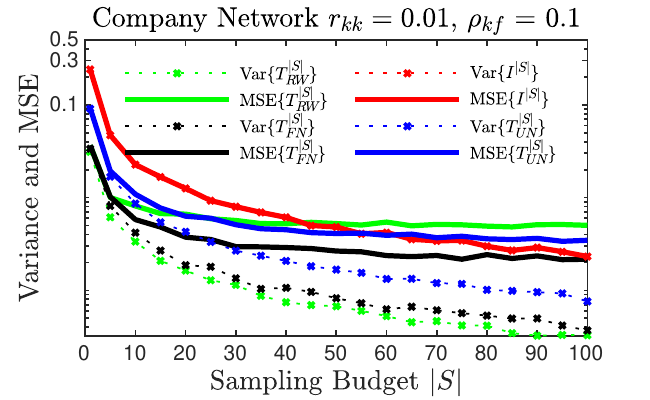}
		\caption{}
	\end{subfigure}
\caption{Empirical MSE and Variance of estimates $\estRW$ (Algorithm~\ref{alg:RW_Sampling}), $\estFN$ (Algorithm~\ref{alg:Friend_of_Node_Sampling}), $\estIP$(intent polling) and $\estUN$ (naive NEP) on five real world datasets (described in Sec.~\ref{subsec:experiments_real_world_SNs}). The subplots show that friendship paradox based NEP methods (Algorithm~\ref{alg:RW_Sampling} and Algorithm~\ref{alg:Friend_of_Node_Sampling})) are more statistically efficient compared to intent polling and naive NEP and, achieves a bias-variance trade-off based on the length of the random walk.}
\label{fig:empirical_results}
\end{figure*}

\subsection{Numerical Examples}
\label{subsec:simulations}

{\bf Generative Models for Graphs: } We use the following two generative models to yield two different types of degree distributions: power-law degree distribution and exponential degree distribution. In all experiments below, we consider graphs with $n = 5000$ nodes. 
\begin{compactitem}
	\item {\bf Configuration Model \cite{molloy1995critical}:} Generate $k$ half-edges for each of the $n$ nodes where $ k \sim ck^{-\alpha}$ (where $c$ is a normalizing constant) and then, connect each half-edge to the another randomly selected half-edge avoiding self loops.  This model yields a power-law degree distribution\footnote{The power-law degree distribution is generally accepted as a key feature of many real world networks such as World Wide Web, Internet and social networks \cite{mislove2007measurement,newman2002random,albert2002statistical, adamic2001search} with a power-law exponent $2<\alpha < 3$\cite{boguna2003epidemic}. Further, it has been shown that friendship paradox and some of its  effects are amplified in the presence of such power-law degree distributions\cite{lerman2016,eom2015tail}. 
	}~$p(k) = ck^{-\alpha}$. We consider two cases: $\alpha = 2.1$ and $\alpha = 2.4$.
	
	\item {\bf Erd\H{o}s-R\'{e}nyi (G(n,p)) model\cite{newman2002random}:} Any two (distinct) nodes are connected by an edge with probability $p$. This model results in a Binomial degree distribution which can be approximated by a Poisson distribution for large $n$. We choose $p = 0.01, n= 5000$ to ensure that the graph has no isolated nodes with high probability. 
\end{compactitem}

\vspace{0.25cm}
\noindent
{\bf Newman's edge-rewiring procedure for modifying neighbor degree correlation:} We utilize the edge-rewiring procedure proposed in \cite{newman2002assortative} to change the assortativity coefficient $r_{kk}$~(\ref{eq:deg_deg_corr}) of the graphs generated using the above models to a desired value while preserving the degree distribution. In the edge-rewiring procedure, two uniformly chosen links  ${(v_1,v_2), (u_1,u_2) \in E}$ at each iteration are replaced with new links $(v_1,u_1), (v_2,u_2)$ if it increases (respectively, decreases) the value of the assortativity coefficient $r_{kk}$. The process is repeated until the desired value of the assortativity coefficient $r_{kk}$ is achieved (or until it no longer changes). 

\vspace{0.25cm}
\noindent
{\bf Simulation Results:} The four polling methods~(Algorithm~\ref{alg:RW_Sampling}, Algorithm~\ref{alg:Friend_of_Node_Sampling}, intent polling~(\ref{eq:intent_polling_estimate}) and naive NEP~(\ref{eq:naive_NEP})) were evaluated on the networks obtained using the simulation setup described above. The MSE of the polling methods were estimated using Monte-Carlo simulation over $600$ independent iterations. The resulting empirical MSE values for the configuration model (power-law degree distribution) are shown in Fig. \ref{fig:mse_pl_alpha_2pt4} and Fig. \ref{fig:mse_pl_alpha_3pt1} for power-law coefficient values ${\alpha = 2.4}$ and ${\alpha = 3.1}$ respectively.  Similarly, results obtained for Erd\H{o}s-R\'{e}nyi graphs (Poisson degree distribution)\footnote{ In the case of Erd\H{o}s-R\'{e}nyi graphs, we only consider assortativity coefficient $r_{kk} = 0$ since it cannot be changed significantly due to the homogeneity in the degree distribution.} are shown in Fig. \ref{fig:mse_ER_davg_50}. The conclusions and insights that can be drawn from these simulation results and how they relate to the analytical results are discussed in Sec.~\ref{sec:discussion}.

\begin{figure*}[]
	\centering
	\begin{subfigure}[!h]{0.3\textwidth}
		\centering
		\includegraphics[width=2.6in]{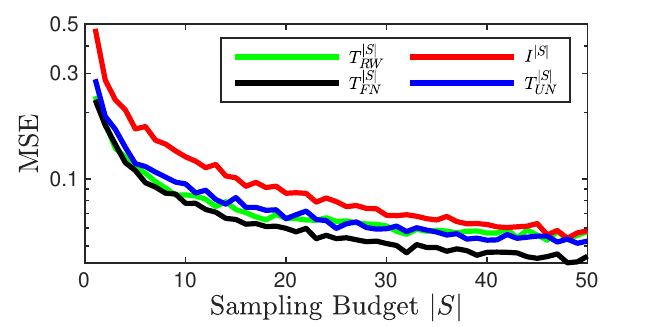}
		\caption{$r_{kk} = 0.2, \rho_{kf} = -0.2$}
		\label{subfig:MSE_alpha2pt4_rkk_0pt2_pkf_neg0pt2}
	\end{subfigure}\hfill
	\begin{subfigure}[!h]{0.3\textwidth}
		\centering
		\includegraphics[width=2.6in]{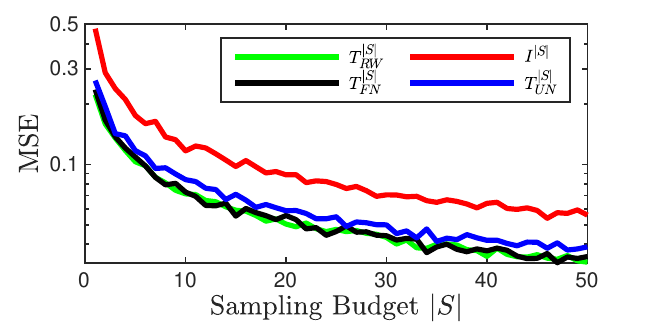}
		\caption{$r_{kk} = 0.2, \rho_{kf} = 0.0$}
		\label{subfig:MSE_alpha2pt4_rkk_0pt2_pkf_0}
	\end{subfigure}\hfill 
	\begin{subfigure}[!h]{0.3\textwidth}
		\centering
		\includegraphics[width=2.6in] {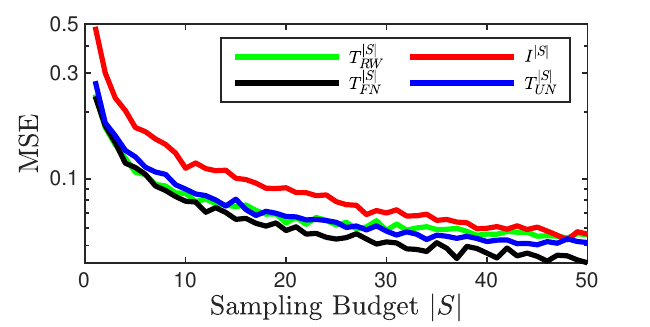}
		\caption{$r_{kk} = 0.2, \rho_{kf} = 0.2$}
		\label{subfig:MSE_alpha2pt4_rkk_0pt2_pkf_0pt2}
	\end{subfigure}

	\begin{subfigure}[!h]{0.3\textwidth}
	\centering
	\includegraphics[width=2.6in]{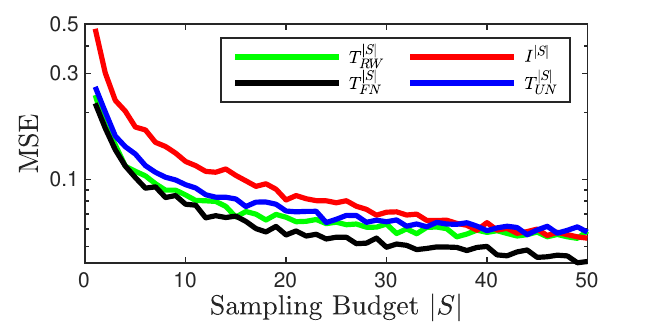}
	\caption{$r_{kk} = 0.0, \rho_{kf} = -0.2$}
	\label{subfig:MSE_alpha2pt4_rkk_0_pkf_neg0pt2}
\end{subfigure}\hfill
\begin{subfigure}[!h]{0.3\textwidth}
	\centering
	\includegraphics[width=2.6in]{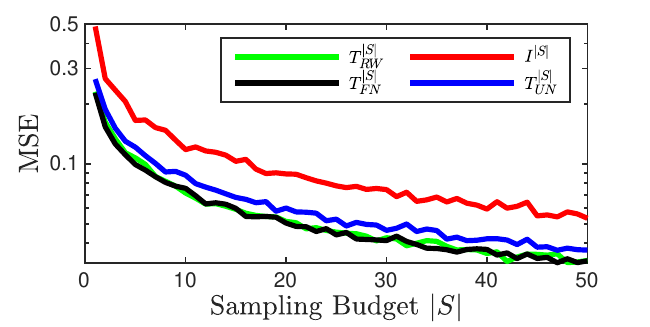}
	\caption{$r_{kk} = 0.0, \rho_{kf} = 0.0$}
	\label{subfig:MSE_alpha2pt4_rkk_0_pkf_0}
\end{subfigure}\hfill 
\begin{subfigure}[!h]{0.3\textwidth}
	\centering
	\includegraphics[width=2.6in] {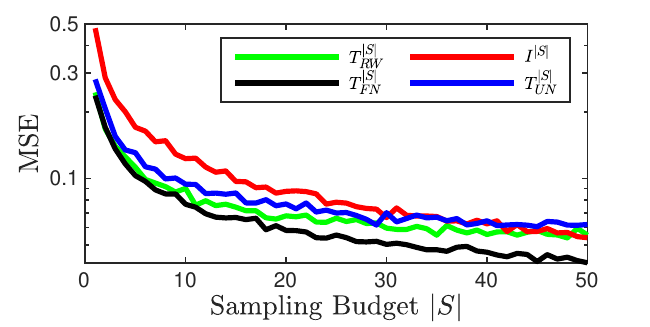}
	\caption{$r_{kk} = 0.0, \rho_{kf} = 0.2$}
	\label{subfig:MSE_alpha2pt4_rkk_0_pkf_0pt2}
\end{subfigure}

	\begin{subfigure}[!h]{0.3\textwidth}
	\centering
	\includegraphics[width=2.6in]{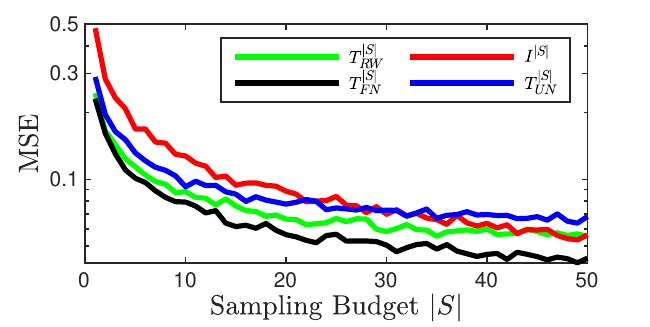}
	\caption{$r_{kk} = -0.2, \rho_{kf} = -0.2$}
	\label{subfig:MSE_alpha2pt4_rkk_neg0pt2_pkf_neg0pt2}
\end{subfigure}\hfill
\begin{subfigure}[!h]{0.3\textwidth}
	\centering
	\includegraphics[width=2.6in]{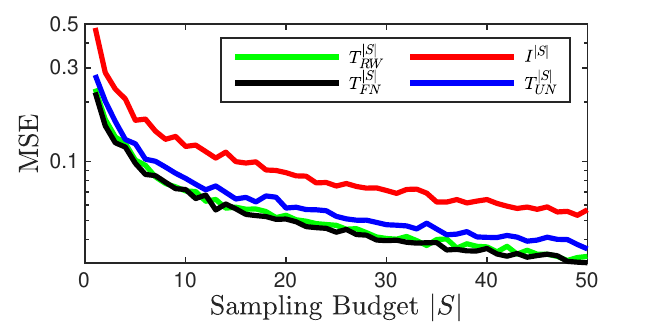}
	\caption{$r_{kk} = -0.2, \rho_{kf} = 0.0$}
	\label{subfig:MSE_alpha2pt4_rkk_neg0pt2_pkf_0}
\end{subfigure}\hfill 
\begin{subfigure}[!h]{0.3\textwidth}
	\centering
	\includegraphics[width=2.6in] {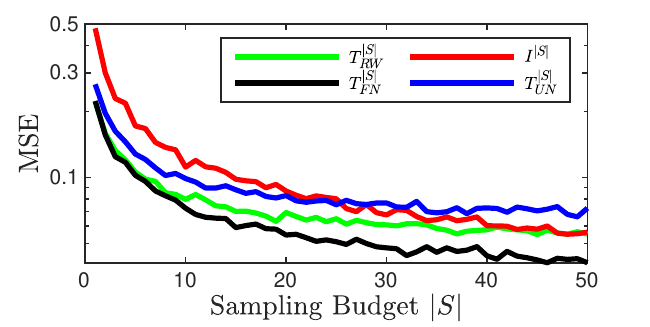}
		\caption{$r_{kk} = -0.2, \rho_{kf} = 0.2$}
	\label{subfig:MSE_alpha2pt4_rkk_neg0pt2_pkf_0pt2}
\end{subfigure}
	\caption{MSE of the estimates obtained using 
		the four polling algorithms
		for a power-law graph with parameter $\alpha = 2.4$ and different values of assortativity coefficient $r_{kk}$ and degree-label correlation coefficient $\rho_{kf}$. Subplots show that, for power-law networks, proposed polling methods have smaller MSE compared to alternative methods under general conditions.}
	\label{fig:mse_pl_alpha_2pt4} 
\end{figure*}

\begin{figure*}[]
	\centering
	\begin{subfigure}[!h]{0.3\textwidth}
		\centering
		\includegraphics[width=2.6in]{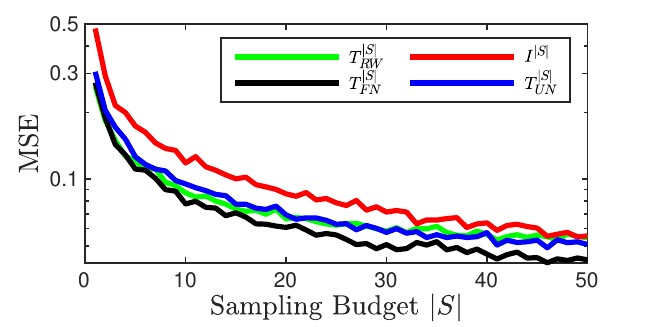}
		\caption{$r_{kk} = 0.2, \rho_{kf} = -0.2$}
		\label{subfig:MSE_alpha3pt1_rkk_0pt2_pkf_neg0pt2}
	\end{subfigure}\hfill
	\begin{subfigure}[!h]{0.3\textwidth}
		\centering
		\includegraphics[width=2.6in]{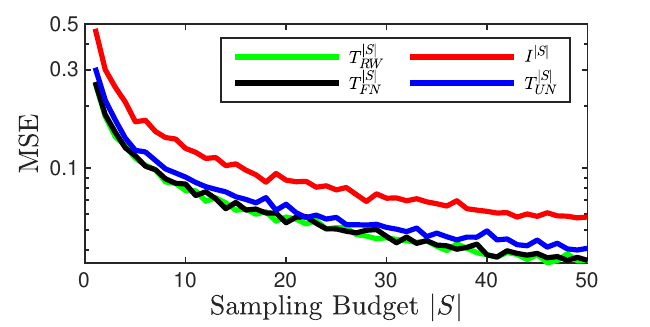}
		\caption{$r_{kk} = 0.2, \rho_{kf} = 0.0$}
		\label{subfig:MSE_alpha3pt1_rkk_0pt2_pkf_0}
	\end{subfigure}\hfill 
	\begin{subfigure}[!h]{0.3\textwidth}
		\centering
		\includegraphics[width=2.6in] {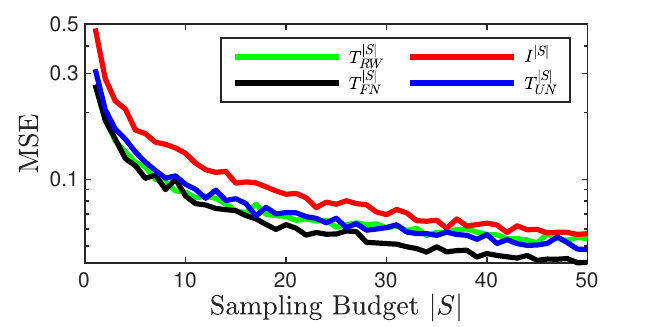}
		\caption{$r_{kk} = 0.2, \rho_{kf} = 0.2$}
		\label{subfig:MSE_alpha3pt1_rkk_0pt2_pkf_0pt2}
	\end{subfigure}
	
	\begin{subfigure}[!h]{0.3\textwidth}
		\centering
		\includegraphics[width=2.6in]{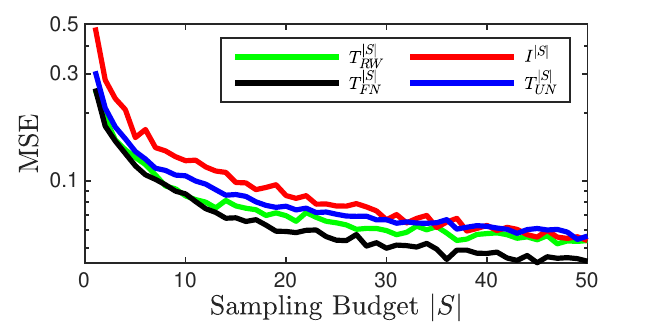}
		\caption{$r_{kk} = 0.0, \rho_{kf} = -0.2$}
		\label{subfig:MSE_alpha3pt1_rkk_0_pkf_neg0pt2}
	\end{subfigure}\hfill
	\begin{subfigure}[!h]{0.3\textwidth}
		\centering
		\includegraphics[width=2.6in]{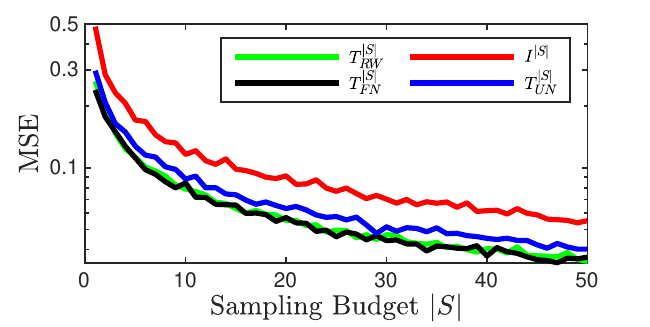}
		\caption{$r_{kk} = 0.0, \rho_{kf} = 0.0$}
		\label{subfig:MSE_alpha3pt1_rkk_0_pkf_0}
	\end{subfigure}\hfill 
	\begin{subfigure}[!h]{0.3\textwidth}
		\centering
		\includegraphics[width=2.6in] {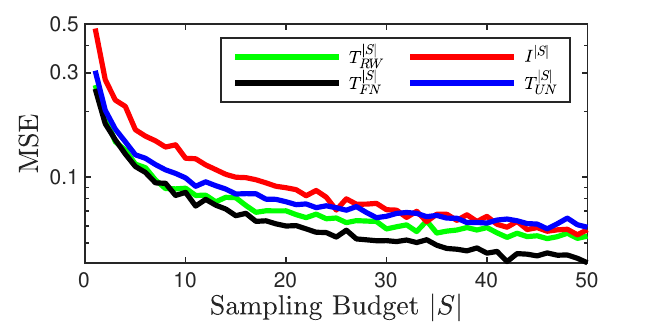}
		\caption{$r_{kk} = 0.0, \rho_{kf} = 0.2$}
		\label{subfig:MSE_alpha3pt1_rkk_0_pkf_0pt2}
	\end{subfigure}
	
	\begin{subfigure}[!h]{0.3\textwidth}
		\centering
		\includegraphics[width=2.6in]{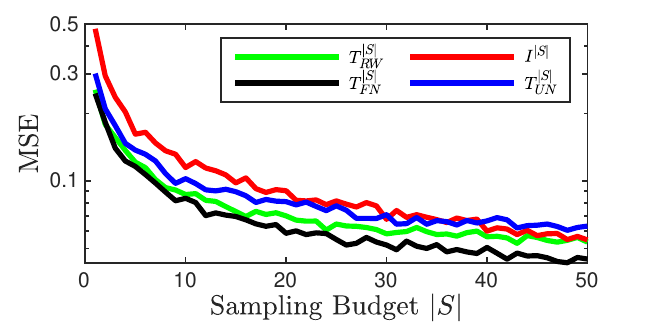}
		\caption{$r_{kk} = -0.1, \rho_{kf} = -0.2$}
		\label{subfig:MSE_alpha3pt1_rkk_neg0pt1_pkf_neg0pt2}
	\end{subfigure}\hfill
	\begin{subfigure}[!h]{0.3\textwidth}
		\centering
		\includegraphics[width=2.6in]{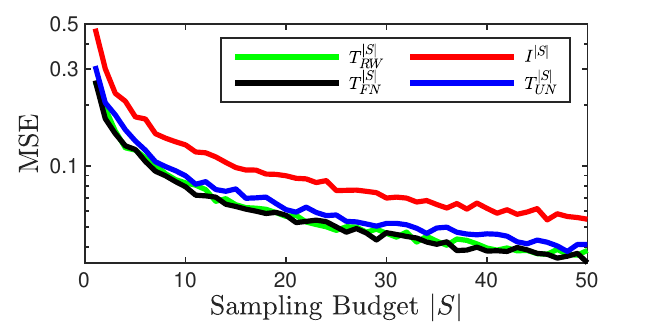}
		\caption{$r_{kk} = -0.1, \rho_{kf} = 0.0$}
		\label{subfig:MSE_alpha3pt1_rkk_neg0pt1_pkf_0}
	\end{subfigure}\hfill 
	\begin{subfigure}[!h]{0.3\textwidth}
		\centering
		\includegraphics[width=2.6in] {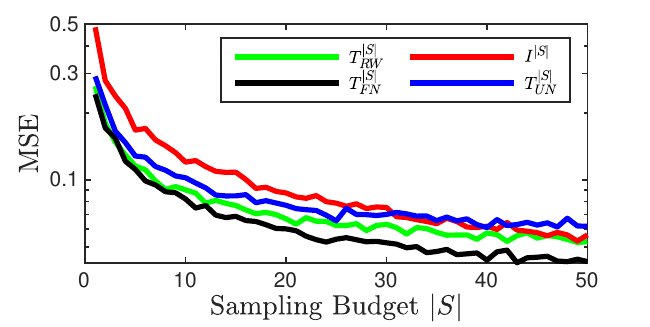}
		\caption{$r_{kk} = -0.1, \rho_{kf} = 0.2$}
		\label{subfig:MSE_alpha3pt1_rkk_neg0pt1_pkf_0pt2}
	\end{subfigure}
	\caption{MSE of the estimates obtained using 
		the four polling algorithms
		for a power-law graph with parameter $\alpha = 3.1$ and different values of assortativity coefficient $r_{kk}$ and degree-label correlation coefficient $\rho_{kf}$. Subplots show that, for power-law networks, proposed polling methods have smaller MSE compared to alternative methods under general conditions.}
	\label{fig:mse_pl_alpha_3pt1} 
\end{figure*}

\begin{figure*}[]
	\centering
	\begin{subfigure}[!h]{0.3\textwidth}
		\centering
		\includegraphics[width=2.6in]{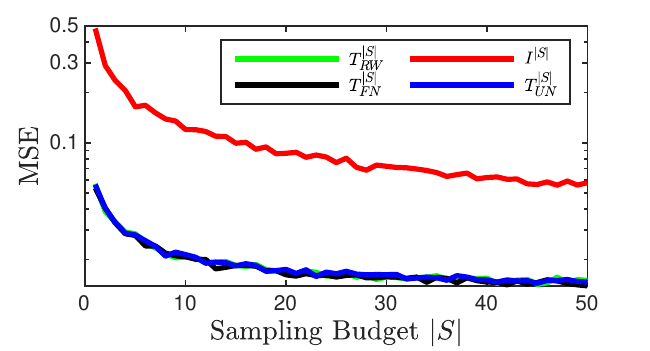}
		\caption{$r_{kk} = 0.0, \rho_{kf} = -0.2$}
		\label{subfig:MSE_davg50_rkk_0_pkf_neg0pt2}
	\end{subfigure}\hfill
	\begin{subfigure}[!h]{0.3\textwidth}
		\centering
		\includegraphics[width=2.6in]{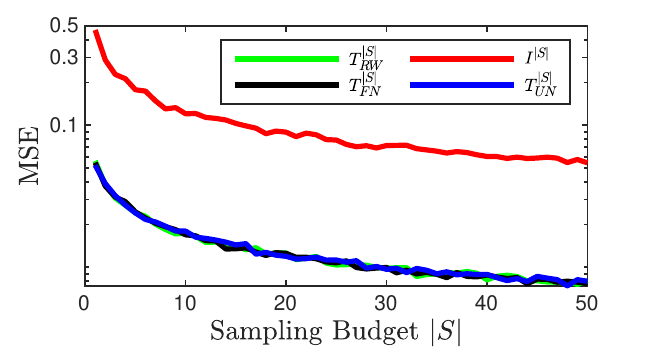}
		\caption{$r_{kk} = 0.0, \rho_{kf} = 0.0$}
		\label{subfig:MSE_davg50_rkk_0_pkf_0}
	\end{subfigure}\hfill 
	\begin{subfigure}[!h]{0.3\textwidth}
		\centering
		\includegraphics[width=2.6in] {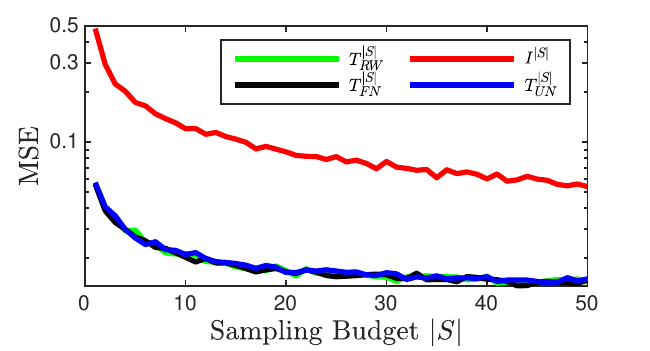}
		\caption{$r_{kk} = 0.0, \rho_{kf} = 0.2$}
		\label{subfig:MSE_davg50_rkk_0_pkf_0pt2}
	\end{subfigure}
	\caption{MSE of the estimates obtained using 
		the four polling algorithms
		for a Erd\H{o}s-R\'{e}nyi~(ER) graph with average degree 50 with assortativity coefficient $r_{kk} = 0$ and different values of degree-label correlation coefficient $\rho_{kf}$. The main conclusion is that, for ER graphs, the proposed friendship paradox based NEP method as well as the greedy deterministic sample selection method result in better performance compared to the intent polling method.}
	\label{fig:mse_ER_davg_50} 
\end{figure*}

\begin{figure*}[!h]
	\centering
	\begin{subfigure}[!h]{0.3\textwidth}
		\centering
		\includegraphics[width=2.6in]{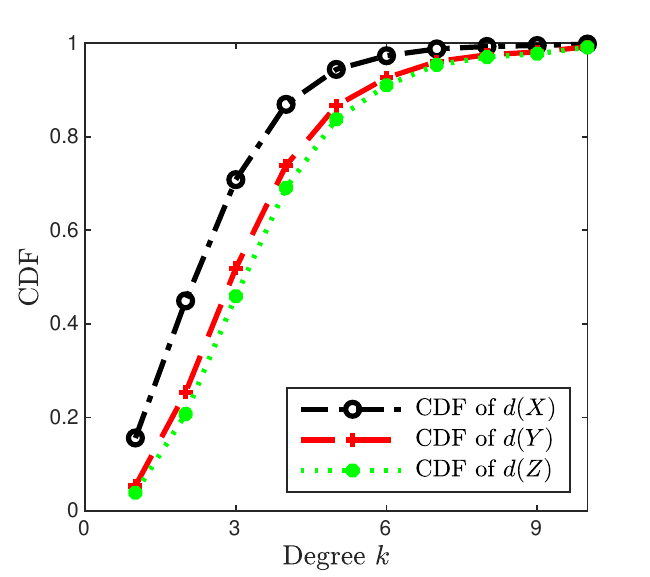}
		\caption{$r_{kk} = -0.2$ (disassortative network)}
		\label{subfig:fosdYZ_alpha2pt4_rkk_neg0pt2}
	\end{subfigure}\hfill
	\begin{subfigure}[!h]{0.3\textwidth}
		\centering
		\includegraphics[width=2.6in]{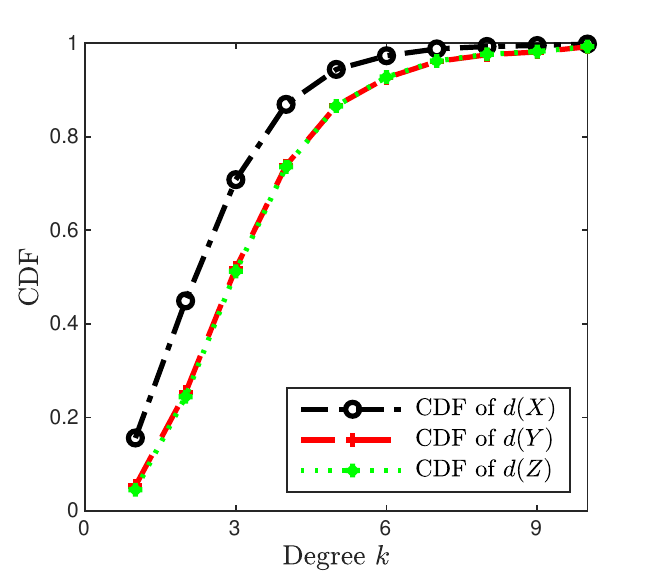}
		\caption{$r_{kk} = 0.0$}
		\label{subfig:fosdYZ_alpha2pt4_rkk_0}
	\end{subfigure}\hfill 
	\begin{subfigure}[!h]{0.3\textwidth}
		\centering
		\includegraphics[width=2.6in] {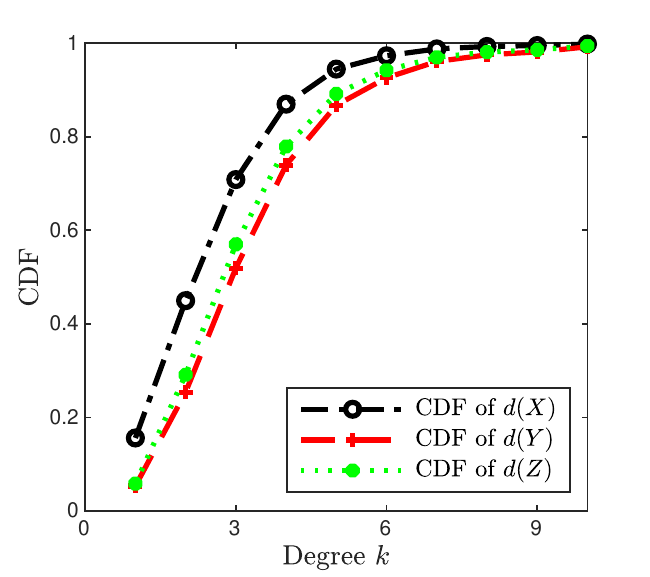}
		\caption{$r_{kk} = 0.2$ (assortative network)}
		\label{subfig:fosdYZ_alpha2pt4_rkk_0pt2}
	\end{subfigure}
	\caption{The cumulative distribution functions (CDF) of the degrees $d(X), d(Y), d(Z)$ of a random node~($X$), a random friend~($Y$) and a random friend of a random node~($Z$)  respectively, for three graphs with the same degree distribution (power-law distribution with a coefficient $\alpha = 2.4$) but different neighbor-degree correlation coefficients $r_{kk}$, generated using the Newman's edge rewiring procedure. This illustrates that $\mathbb{E} \{d(Z)\} \geq \mathbb{E}\{d(Y)\}$ for $r_{kk} \leq 0$ (Fig. \ref{subfig:fosdYZ_alpha2pt4_rkk_neg0pt2}) and vice-versa. This figure also shows how the distributions of $d(X), d(Y)$ remain invariant to the changes in the joint degree distribution $e(k,k')$ that preserve the degree distribution $P(k)$.}
	\label{fig:fosdYZ_alpha2pt4} 
\end{figure*}

\section{Discussion of Empirical and Simulation Results}
\label{sec:discussion}
This section discusses the insights and conclusions that can be drawn from the empirical and simulation results~(Sec.~\ref{sec:experiments}) and, how they relate to the analytical results~(Sec.~\ref{sec:analysis}). The main aim is to highlight how the analytical and experimental results help to identify the contexts for which each polling algorithm is suitable. 

\subsection{Power-law Graphs}
\noindent
{\bf Intent Polling vs. Friendship Paradox Based NEP: } Corollary~\ref{cor:b_range} stated that the friendship paradox based NEP Algorithm~\ref{alg:RW_Sampling} outperforms the classical intent polling in terms of the mean-squared error for small sampling budget ${\samplingbudget}$ values. The empirical results (Fig.~\ref{fig:empirical_results}) are consistent with Corollary~\ref{cor:b_range}; it can be seen that the MSE of the intent polling estimate $\estIP$ is larger than the MSE of the estimates $\estRW, \estFN$ obtained via the friendship paradox based NEP methods for smaller (less than $50$) sampling budget ${\samplingbudget}$ values. Further, MSE of estimates $\estRW, \estFN$ are smaller for all considered sampling budget ${\samplingbudget}$ values when the degree-label correlation coefficient $\rho_{kf}$ is zero (and hence, the friendship paradox based polling produces an unbiased estimate according to Theorem~\ref{th:bias_var_Trw}). Hence, both analytical and empirical results indicate that friendship paradox based NEP methods outperform the classical intent polling method when the sampling budget ${\samplingbudget}$ is constrained to be smaller or, the node labels are uncorrelated with the node degrees ($\rho_{kf} = 0$). 

\vspace{0.25cm}
\noindent
{\bf Effect of degree-label correlation ($\rho_{kf}$): } Fig.~\ref{fig:empirical_results} shows that the friendship paradox based polling Algorithms~\ref{alg:RW_Sampling}~and~\ref{alg:Friend_of_Node_Sampling} outperform both intent polling and naive NEP~(\ref{eq:naive_NEP}) for all considered sampling budget ${\samplingbudget}$ values when the node labels and node degrees are uncorrelated ($\rho_{kf} = 0$). 
When the node degrees and node labels are correlated ($\rho_{kf} \neq 0$), Algorithm~\ref{alg:Friend_of_Node_Sampling} still outperforms (in terms of MSE) both intent polling and naive NEP methods for all considered sampling budget ${\samplingbudget}$ values whereas naive NEP method outperforms Algorithm~\ref{alg:RW_Sampling} when ${\samplingbudget}$ becomes large due to the bias variance trade-off that is discussed next. 

\vspace{0.25cm}
\noindent
{\bf Friendship paradox based bias variance trade-off optimization: }
Note that the naive NEP estimate $\estUN$, NEP estimate $\estFN$ based on version 2 of friendship paradox (Theorem~\ref{th:fosd_X_Z}) and NEP estimate $\estRW$ based on version~1 of friendship paradox (Theorem~\ref{th:friendship_paradox_Feld}) correspond to random walks of length $N=0$ ($\estUN$), $N=1$ ($\estFN$) and $N\rightarrow \infty$ ($\estRW$). As such, $\estUN$ is based on responses of individuals sampled independent of their degree, $\estRW$ is based on responses of individuals sampled with probabilities proportional to their degrees and $\estFN$ achieves a trade-off by taking only a single step random walk. Therefore, it is intuitive that the variance of the estimates should satisfy ${\var\{\estRW\} \leq \var\{\estFN\}  \leq \var\{\estUN\}}$, agreeing with Corollary~\ref{cor:var_upperbounds_comparisn_Tun_Trw} and the empirical variances plotted in Fig.~\ref{fig:empirical_results}. However,  in terms of the mean-squared error (which takes the bias of the estimates also into account), $\estFN$ outperforms both $\estUN, \estRW$ (in terms of MSE) for all ${\samplingbudget}$ values considered in the empirical results. This observation suggests that the length of random walk (e.g. $N=1$ in the case of estimate $\estFN$) can be used to control the bias-variance trade-off of the friendship paradox based NEP methods. For example, if it is apriori known to the pollster that the labels have negligible correlation with the degrees (i.e. $\rho_{kf} \approx 0$ and hence, the bias of both $\estRW, \estFN$ will be negligible), she can choose to use $\estRW$ to minimize the variance of the estimate.

\vspace{0.25cm}
\noindent
{\bf Effect of the heavy-tails: } Comparing Fig.~\ref{fig:mse_pl_alpha_2pt4} with Fig.~\ref{fig:mse_pl_alpha_3pt1} shows that the MSE of Algorithm \ref{alg:RW_Sampling} and Algorithm \ref{alg:Friend_of_Node_Sampling} are smaller in the network with power-law coefficient~$\alpha = 2.1$ compared to that with~$\alpha = 2.4$. The difference in MSE is more pronounced for Algorithm~\ref{alg:Friend_of_Node_Sampling} compared to Algorithm~\ref{alg:RW_Sampling}. This suggests that friendship paradox based algorithms are more suitable when the underlying network has a heavy tailed degree distribution.  

\vspace{0.25cm}
\noindent
{\bf Effect of the Assortativity of the Network:} 
Different joint degree distributions $e(k,k')$ can yield  the same neighbor degree distribution $q(k)$ (explained in Sec.~\ref{sec:experiments}). Naturally, this marginal distribution $q(k)$ does not capture the joint variation of the degrees a random pair of neighbors. In Algorithm \ref{alg:RW_Sampling} (which samples neighbors uniformly), the degree distribution of the samples (i.e. queried nodes) is the neighbor degree distribution $q(k)$. Hence, the performance is not affected by the assortativity coefficient $r_{kk}$, which captures the joint variation of the degrees of a random pair of neighbors. This is seen in Fig. \ref{fig:mse_pl_alpha_2pt4} where, each column (corresponding to different $r_{kk}$ values) has approximately same MSE for Algorithm \ref{alg:RW_Sampling}. However, the MSE of Algorithm~\ref{alg:Friend_of_Node_Sampling} (that samples random friends $Z$ of random nodes) increases with assortativity $r_{kk}$ due to the fact that the distribution of degree $d(Z)$ of a random friend $Z$ of a random node is a function of the joint degree distribution. In order to highlight this further, Fig.~\ref{fig:fosdYZ_alpha2pt4} illustrates the effect of the neighbor degree correlation $r_{kk}$ on the distribution of $d(Z)$ (and the invariance of the distribution of $d(Y)$ to $r_{kk}$). This result indicates that, if the network is disassortative~($r_{kk} < 0$), Algorithm~\ref{alg:Friend_of_Node_Sampling} is a more suitable choice for polling compared to Algorithm~\ref{alg:RW_Sampling}.

\subsection{Erd\H{o}s-R\'{e}nyi Graphs}
The Erd\H{o}s-R\'{e}nyi ($G(n, p)$) model constructs a random graph as follows: start with $n$ vertices and then connect any two vertices  with probability $p$. Therefore, the average degree of the resulting graph is~$(n-1)p$. From~Fig.~\ref{fig:mse_ER_davg_50}, it can be seen that both Algorithm \ref{alg:RW_Sampling} and Algorithm \ref{alg:Friend_of_Node_Sampling} yield a smaller MSE than the intent polling method for an Erd\H{o}s-R\'{e}nyi network with $p=0.01$ and $n = 5000$. Also, Algorithm~\ref{alg:RW_Sampling} and Algorithm~\ref{alg:Friend_of_Node_Sampling} have approximately equal MSE. This is due to the fact that in an  Erd\H{o}s-R\'{e}nyi network, the neighbor degree correlation is approximately zero and therefore, distributions of the degree~$d(Y)$ of a random neighbor~$Y$ and the distribution of the degree~$d(Z)$ of a random neighbor~$Z$ of a random node are approximately equal.

\section{Conclusion}

This paper considered the problem of estimating the fraction of nodes in a graph that has a particular attribute (represented by a binary label~$1$) and, proposed a novel class of polling methods called Neighborhood Expectation Polling~(NEP). In NEP, each sampled individual responds with information about the fraction of her neighbors in the social network that has label~$1$. We considered the cases where, either: 1)~the pollster has no knowledge about the social graph but, has the ability to perform random walks on the graph 2)~uniformly sampled nodes from the unknown social graph are available. Two NEP algorithms were proposed (for case~1 and case~2) exploiting a form of network bias called ``friendship paradox". Theorems~\ref{th:NEP_with_vs_without_FP}~to~\ref{th:bias_var_Tfn} characterized the bias, variance and mean-squared error of the estimates~(obtained via the proposed algorithms) as well as how they depend on the properties of the underlying network (correlation between node labels and degree, expansion and, average, minimum and maximum degree, etc.). These results are useful for a pollster to incorporate prior knowledge about the underlying network to choose the best algorithm (in terms of statistical efficiency) and guarantee its performance.  Extensive empirical and simulation results are provided to illustrate the performance of the proposed methods under different network properties. These complement the theoretical analysis and provide insights into how the proposed algorithms would perform under different conditions. Both theoretical and experimental results indicate that the friendship paradox based NEP algorithms are capable of obtaining an estimate with a smaller mean-squared error with only a smaller (compared to alternative methods) number of respondents. 

%

%
%
%
%
%
%

\ifCLASSOPTIONcompsoc
  \section*{Acknowledgments}
\else
  \section*{Acknowledgment}
\fi
The authors thank Jon Kleinberg~(Department of Computer Science of Cornell University) and Rahmtin Rotabi~(Google) for helpful suggestions. 

\ifCLASSOPTIONcaptionsoff
  \newpage
\fi

\end{document}